\documentclass[10pt,a4paper]{article}
\usepackage{amsmath,amsfonts,amsthm,amssymb}
\usepackage[top=3cm,left=3cm,right=2.5cm,bottom=2.5cm]{geometry}

\usepackage{graphicx}
\usepackage{import}
\usepackage[dvipsnames]{xcolor}
\definecolor{myblu}{HTML}{2b4781}
\definecolor{myor}{HTML}{dd7008}
\usepackage[colorlinks,citecolor=myblu,urlcolor=myblu,linkcolor=myblu]{hyperref}

\numberwithin{equation}{section}

\newtheorem{corollary}{Corollary}[section]
\newtheorem{proposition}[corollary]{Proposition}
\newtheorem{lemma}[corollary]{Lemma}
\newtheorem{theorem}[corollary]{Theorem}
\newtheorem{remark}[corollary]{Remark}
\newtheorem{example}[corollary]{Example}
\newtheorem{definition}[corollary]{Definition}

\newcommand{\numset}[1]{\mathbf{#1}}
	\newcommand{\cc}{\numset{C}}
	\newcommand{\rr}{\numset{R}}
	
	\newcommand{\nn}{\numset{N}}

	\usepackage{braket}
	\newcommand{\one}{\mathbf{1}}

		\newcommand{\Exp}[1]{\mathrm{e}^{#1}}
	
	\renewcommand{\Re}{\operatorname{Re}}
	\newcommand{\vol}{\operatorname{vol}}
	\newcommand{\ee}{\mathbf{E}}
	\newcommand{\pp}{\mathbf{P}}
	\makeatletter
	\providecommand*{\diff}%
		{\@ifnextchar^{\DIfF}{\DIfF^{}}}
	\def\DIfF^#1{%
		\mathop{\mathrm{\mathstrut d}}%
			\nolimits^{#1}\gobblespace}
	\def\gobblespace{%
		\futurelet\diffarg\opspace}
	\def\opspace{%
		\let\DiffSpace\!%
		\ifx\diffarg(%
			\let\DiffSpace\relax
		\else
			\ifx\diffarg[%
				\let\DiffSpace\relax
			\else
				\ifx\diffarg\{%
					\let\DiffSpace\relax
				\fi\fi\fi\DiffSpace}
	
	\renewcommand{\d}{\diff}

	\renewcommand{\sp}{\operatorname{sp}}
	\DeclareMathOperator{\spb}{spb}
	
	\DeclareMathOperator{\supp}{supp}
	
		\newcommand{\cD}{\mathsf{D}}

	\DeclareMathOperator{\cl}{cl}
	\DeclareMathOperator{\interior}{int}
	\DeclareMathOperator{\tr}{tr}

	\renewcommand{\ss}{\mathcal{S}}
	\newcommand{\cP}{\mathcal{P}}
	\newcommand{\cQ}{\mathcal{Q}}
	\newcommand{\cCt}{\mathcal{C}_t}
	\newcommand{\cB}{\mathcal{B}}
	\newcommand{\cA}{\mathcal{A}}
	\newcommand{\g}{g}

	\providecommand{\upDelta}{\Delta}
	\providecommand{\upGamma}{\Gamma}

\begin{document}

\title{The large-time and vanishing-noise limits \\ for entropy production in nondegenerate diffusions}
\date{}
\author{Renaud Raqu\'epas}

\maketitle

{\begin{center}
		\small
		\begin{tabular}{c}
			 New York University  \\
			 Dept.\ of Mathematics at the Courant Institute \\
			 251 Mercer Street \\
			 New York, NY 10012, United States
		\end{tabular}
	\end{center}
}

\begin{abstract}
  We investigate the behaviour of a family of entropy production functionals associated to stochastic differential equations of the form
	\[
		\d X_s = -\nabla V(X_s) \d s + b(X_s) \d s + \sqrt{2\epsilon} \d W_s,
	\]
	where~$b$ is a globally Lipschitz nonconservative vector field keeping the system out of equilibrium, with emphasis on the large-time limit and then the vanishing-noise limit. Different members of the family correspond to different choices of boundary terms. Our analysis yields a law of large numbers and a local large deviation principle which does not depend on the choice of boundary terms and which exhibits a Gallavotti--Cohen symmetry. We use techniques from the theory of semigroups and from semiclassical analysis to reduce the description of the asymptotic behaviour of the functional to the study of the leading eigenvalue of a quadratic approximation of a deformation of the infinitesimal generator near critical points of~$V$. \\


\end{abstract}

\tableofcontents

\section{Introduction}

The study of reversibility of diffusion processes was pioneered by A.N.~Kolmogorov in~\cite{Ko37}, with one of its first basic results being that a diffusion
\begin{equation}\label{eq:nd-diff}
	\d X_s = c(X_s) \d s+ \sigma \d W_s
\end{equation}
on~$\rr^N$ with constant diffusion matrix~$\sigma > 0$ and initial condition~$X_0 \sim \lambda$ is reversible if and only if there exists a function~$U$ such that $c = -\sigma \sigma^\mathsf{T} \nabla U $ and $\lambda$ is the unique probability measure whose density is proportional to $\exp(-2 U)$.
In all other cases, the time reversal of the original diffusion is a Markov process which is different from the original one.

The question whether the reversal of a diffusion is itself a diffusion was explored and understood in the 1980s, most notably by B.\,D.~O.~Anderson~\cite{An82} and by E.~Pardoux and U.~Haussmann~\cite{HP86}. When it is the case, it is natural to ask how distinguishable the two diffusions are: this more quantitative question\,---\,and its connection with thermodynamics\,---\,is the subject of the present paper. It has a long history in both the physics and mathematics literature, but we will only give references to the mathematically rigorous works on the particular aspects we are interested in.

Both the original process observed during the interval~$[0,t]$ and its time reversal give rise to probability measures on a space of continuous functions (the trajectories, or paths); let us call them respectively~$\cP_t$ and~$\cP_t \circ \Theta_t^{-1}$. Using statistical tools to distinguish between these two measures is called \emph{hypothesis testing of the arrow of time} in~\cite{JOPS12,EPbrief}. To explore the basic questions in the realm of hypothesis testing, the log-likelihood ratio
\begin{equation}\label{eq:LLR}
	S_t^{\textnormal{LLR}} := \log \frac{\d \cP_t}{\d (\cP_t \circ \Theta_t^{-1})}
\end{equation}
and its moment-generating function are of great significance;~$S_t^{\textnormal{LLR}}$ is sometimes called the \emph{canonical entropy production functional}.

In dimension~2 or~3, the diffusion~\eqref{eq:nd-diff}\,---\,sometimes called an overdamped Langevin equation\,---\,is naturally interpreted as a small-inertia approximation of the dynamics of a single particle in the force field~$c$, perturbed by a thermal noise; the matrix $\sigma\sigma^{\mathsf{T}}$ is related to temperature through an Einstein-type relation.
Hence, a thermodynamical notion of entropy production is natural: with $b$ the part of~$c$ that is nonconservative, the integral
\begin{equation}\label{eq:NCW}
	S_t^{\textnormal{W}} := 2 \int_0^t \big\langle \big(\sigma\sigma^\mathsf{T}\big)^{-1} b(X_s), \circ \d X_s \big\rangle
\end{equation}
is the work done by the nonconservative force, appropriately rescaled by the temperature~\cite{Ku98,LS99}. When considering several particles,~$\sigma\sigma^{\mathsf{T}}$ may have different blocks for different particles, each related to a (possibly different) temperature, and the above integral can be split into a sum of the corresponding contributions~\cite{LS99,MNV03}.

One expects~$S_t^{\textnormal{LLR}}$ and~$S_t^{\textnormal{W}}$ to be quantities of order~$t$ and to only differ by an additive term that depends on the initial and final conditions of the paths. In the present article, we consider an abstract entropy production functional~$\mathcal{S}_t$ corresponding to any sufficiently well-behaved modification of these so-called \emph{boundary terms} and study its behaviour as~$t \to \infty$ and then in the limit as~$\sigma$ vanishes.
The way in which we take~$\sigma$ to~$0$ leaves out some geometric considerations: we consider $\sigma^\epsilon = \sqrt{2\epsilon}\one$ and take the scalar parameter~$\epsilon$ to~$0$. Considering the more general case~$\sigma^\epsilon = \sqrt{\epsilon}\sigma^1$ for some fixed positive-definite matrix~$\sigma^1$ has been sacrificed for readability and ease of interpretation of the formulas:
one can perform a suitable change of variables and carry on with a similar analysis, but one must then be careful with the physical interpretation. Indeed, from the physical point of view, the case we look at here corresponds to situations where the lack of equilibrium comes from a nonconservative driving force and is conceptually different from situations where the lack of equilibrium comes from an imbalance between the sources of thermal fluctuations.

The asymptotic behaviour of entropy production functionals as~$t \to \infty$ at fixed~$\epsilon > 0$ was studied by L.~Bertini and G.~Di Ges\`u in~\cite{BD15} and by F.Y.~Wang, J.~Xiong and L.~Xu in~\cite{WXX16} under more restrictive technical conditions. For a class of \emph{degenerate} linear diffusions, V.~Jak\v{s}i\'c, C.-A.~Pillet and A.~Shirikyan have performed a very detailed analysis of the limit~$t \to \infty$~\cite{JPS17}.
In~\cite{BD15}, the authors also tackled the rescaled limit $\epsilon \to 0$ at fixed~$t > 0$ by means of Freidlin--Wentzell theory and then the limit as $t \to \infty$ using subadditivity and results on $\upGamma$-convergence; also see~\cite{RT00} and~\cite{Ku07}.
As already discussed by some of these authors,
taking~$t \to \infty$ \emph{first} and~\emph{then} $\epsilon \to 0$ is physically more natural and was left open. This order is the one taken here, revealing a different picture than in~\cite{BD15}.

In Section~\ref{sec:setup}, we set our assumptions, discuss the basic theory surrounding the time reversal of the diffusion
and rigorously relate~\eqref{eq:LLR} to~\eqref{eq:NCW} by boundary terms.
In Section~\ref{sec:EPF}, we introduce a family of entropic functionals~$\ss_t^\epsilon$ depending on the choice of boundary terms. We then give a representation for the moment-generating function~$\chi_t^\epsilon(\alpha)$ involving the chosen boundary term and the compact and irreducible semigroup generated by the deformation
\[
\Lambda^{\epsilon,\alpha} = \epsilon \upDelta +  \braket{c - 2\alpha b,\nabla } - \tfrac{\alpha(1-\alpha)}{\epsilon} |b|^2  + \tfrac {\alpha}{\epsilon} \braket{b,b - c} - \alpha\operatorname{div} b
\]
of the generator associated to~\eqref{eq:nd-diff}. Relevant spectral properties of~$\Lambda^{\alpha,\epsilon}$, including domain technicalities, a Perron--Frobenius-type result for its spectral bound, denoted~$e^\epsilon(\alpha)$, and a result of convergence in the large-time limit for the generated semigroup, are given in Appendix~\ref{app:semig}.

In Section~\ref{sec:large-t}, we study the asymptotics of the moment-generating function as~$t \to \infty$ for fixed~$\epsilon > 0$: we show that
\[
	\lim_{t\to\infty}\frac 1t\log \chi_t^\epsilon(\alpha) = e^\epsilon(\alpha)
\]
for a set of~$\alpha$ which depends on the behaviour of the boundary terms at infinity.
Our set of assumptions is more general than that of~\cite{BD15}: we most notably allow~$b$ to be unbounded; see Assumptions~\textnormal{(L0)},~\textnormal{(L1)} and~\textnormal{(RB)}.
{Yet, they still allow us to prove our \emph{first main result}: if the behaviour at infinity of the boundary terms is suitable\,---\,as made precise by Assumption~\textnormal{(IP)}\,---, a local large deviation principle (\textsc{ldp}) holds. This \textsc{ldp} is stated precisely as Proposition~\ref{prop:ldp-t} but can be roughly summarized as the validity of the asymptotics
\begin{equation}
\label{eq:heur-ldp}
		\pp \{ t^{-1}\ss_t^\epsilon \approx \varsigma \} \asymp \exp\left(-t e^\epsilon_*(\varsigma)\right)
\end{equation}
for $0 < \epsilon^{-1} \ll t$ and} all $\varsigma$ close enough to the mean,
where the rate function~$e^\epsilon_*$ is the Legendre transform of~$e^\epsilon$.
{We emphasize that this local \textsc{ldp} is common to all the boundary terms that satisfy our assumptions, including the two natural choices presented earlier\,---\,this is rare in the existing literature. Since it is known that these choices may give rise to different behaviour of the rate functions far away from the mean~\cite{CVZ03,JPS17}, this shows that\,---\,even though we have not focused on enlarging the interval of validity of the principle as much as technically possible\,---\,the local nature is not merely technical.} In Section~\ref{sec:lin}, we characterize the vanishing of the mean entropy production per unit time and give detailed information on the rate function~$e^*_\epsilon$ in the case where the diffusion is linear.

In Section~\ref{sec:small-e}, we use the linear case and a result in semiclassical analysis proved in Section~\ref{app:semiclass} to describe the asymptotic behaviour of the rate function~$e_*^\epsilon$ as~$\epsilon \to 0$ in the more general case covered by our assumptions.
This gives our \emph{second main result}, Theorem~\ref{prop:ldp-t-eps}. Summarized roughly, it states the validity of the asymptotics
\begin{equation}
\label{eq:heur-ldp-van}
		\pp \{ t^{-1}\ss_t^\epsilon \approx \varsigma \} \asymp \exp\left(-t e_*(\varsigma)\right)
\end{equation}
for $1 \ll \epsilon^{-1} \ll t$ and all $\varsigma$ close enough to the mean,
for some limiting rate function~$e_*$. A key point here is that the limit~$e_*$ is easily numerically accessible through simple matrix equations and that, analytically, many properties can be deduced from the behaviour of the noiseless dynamics near the critical points of~$V$. However, this analysis of the vanishing-noise limit requires extra conditions on the behaviour of the vector fields near the critical points of~$V$; see Assumption~\textnormal{(ND)}.

\paragraph*{Acknowledgements.}
The work presented in this article was done while the author was a student at McGill University (Dept.\ of Mathematics and Statistics) and at Univ.\ Grenoble Alpes (Institut Fourier).
During this period, the research of the author was partially supported by the Natural Sciences and Engineering Research Council of Canada and by the {Agence Nationale de la Recherche} through the grant NonStops (ANR-17-CE40-0006). The author wishes to thank Vojkan Jak\v{s}i\'c and Armen Shirikyan for guidance through the early stages of this project, as well as No\'e Cuneo, Alain Joye and Claude-Alain Pillet for comments on earlier versions of this work.

\section{Setup, definitions and preliminary results}
\label{sec:setup}

We study a stochastic differential equation (\textsc{sde}) in~$\rr^N$ of the form
\begin{equation*}
	\d X_s = -\nabla V(X_s) \d s + b(X_s) \d s + \sqrt{2\epsilon} \d W_s,
\end{equation*}
where~$V$ is a coercive Morse function and~$b$ is a nonconservative vector field vanishing at the critical points of~$V$, and the log-likelihood ratio~\eqref{eq:LLR} between the corresponding path measure and its time reversal. We will explicitly keep track of the dependence on the initial condition and on~$\epsilon$ as superscripts for relevant quantities.

\begin{remark}\label{rem:orth-decomp}
	There is some freedom in decomposing a deterministic drift in the form~$-\nabla V + b$. Because this drift may already be provided in a given such decomposition coming from a physical context, we facilitate the verification of our hypotheses by avoiding making the assumption that this decomposition is in any sense canonical.
\end{remark}

\subsection{Assumptions on the equation and immediate consequences}

Throughout the paper, $N \geq 2$ is a fixed natural number and the $N$-dimensional euclidean space~$\rr^N$ is equipped with the standard inner product~$\braket{\,\cdot\,,\cdot\,}$. Let $V : \rr^N \to \rr$ be a fixed function of class~$C^3$ and $b : \rr^N \to \rr^N$ a fixed globally Lipschitz vector field of class~$C^2$.
We introduce the following assumptions.
\begin{description}
	\item[Assumption (L0).] There exists a positive-definite matrix~$H_0$ and a constant~$K_0$ such that
	\[
		\braket{\nabla V (x),H_0 x} \geq |x|^2 - K_0
	\]
	for all~$x \in \rr^N$ and the function $x \mapsto |\nabla V(x)|^2 - a \|D^2 V(x)\|$ is bounded  below for all values of~$a \in \rr$.
	\item[Assumption (L1).] There exists a positive-definite matrix $H_b$ and a constant~$K_b$ such that
	\begin{equation*}
		\braket{\nabla V(x) - b(x), H_b x} \geq |x|^2 - K_b
		\label{eq:coe}
	\end{equation*}
	for all~$x \in \rr^N$.
	\item[Assumption (RB).] There exist constants $h_b \in (0,\infty)$ and $k_b \in [0,\tfrac 12)$ such that
	\begin{equation*}
		\braket{b(x),\nabla V(x)} \leq k_b|\nabla V(x)|^2
	\qquad
	\textnormal{ and }
	\qquad
		|b(x)|^2 \leq h_b |\nabla V(x)|^2
	\end{equation*}
	for all~$x \in \rr^N$.
	\item[Assumption (ND).] The critical points of~$V$ form a finite set~$\{x_j\}_{j=1}^m$ and
	\begin{equation*}
		\det D^2V|_{x_j} \neq 0
	\end{equation*}
	for each~$j=1,\dotsc,m$.
\end{description}

\medskip

Assumption~\textnormal{(L0)} yields a Lyapunov structure for the ordinary differential equation
\begin{equation*}
	\dot Y = -\nabla V(Y)
\end{equation*}
and Assumption~\textnormal{(L1)} plays the same role for
\begin{equation*}
	\dot X = -\nabla V(X) + b(X).
\end{equation*}
The relative bounds in Assumption~\textnormal{(RB)} guarantee that these two deterministic dynamics have the same fixed points, which form a finite set~$\{x_j\}_{j \in \nn}$ and are all nondegenerate by~\textnormal{(ND)}.
The regularity assumptions on~$V$ and~$b$ are made for simplicity of technical estimates and can be relaxed if necessary. For example, in the case of~$V$, class~$C^2$ with~$x \mapsto D^2V|_{x}$ locally H\"older continuous would only require minor changes to the proofs in Section~\ref{app:semiclass}.

The Lyapunov structure for the ordinary differential equations guarantees the existence and uniqueness of the solutions to the \textsc{sde}s
\begin{equation}
\left\{
\begin{array}{r l}
	\d Y^{y,\epsilon}_s &\!= -\nabla V(Y^{y,\epsilon}_s) \d s + \sqrt{2\epsilon} \d W_s, \\
	Y^{y,\epsilon}_0 &\!= y,
\end{array}
\right.
\label{eq:SDE-Y}
\end{equation}
and
\begin{equation}
\left\{
\begin{array}{r l}
	\d X^{x,\epsilon}_s &\!= -\nabla V(X^{x,\epsilon}_s) \d s + b(X^{x,\epsilon}_s) \d s + \sqrt{2\epsilon} \d W_s, \\
	X^{x,\epsilon}_0 &\!= x.
\end{array}
\right.
\label{eq:SDE}
\end{equation}
The study of~\eqref{eq:SDE} is intimately related to partial differential equations involving
\[
	\Lambda^{\epsilon,0} := \epsilon\upDelta + \braket{-\nabla V + b, \nabla}
\]
and its adjoint. We refer to~\cite[\S{3.6}]{Has} 
for the general basic relations and to Appendix~\ref{app:semig} for more precise technical properties of~$\Lambda^{\epsilon,0}$ in this specific case. {Basic probabilistic properties of the solutions of~\eqref{eq:SDE} are provided in Appendix~\ref{app:cons}.}
{In particular,} there exists a unique stationary measure~$\lambda^\epsilon_\textnormal{inv}$ for~\eqref{eq:SDE};~$\mu_0^\epsilon$ for~\eqref{eq:SDE-Y}. Both~$\lambda^\epsilon_\textnormal{inv}$ and~$\mu_0^\epsilon$ possess positive continuous densities with respect to the Lebesgue measure on~$\rr^N$, denoted~``vol'' hereafter. Whenever we write ``almost everywhere'' or ``almost all'' without specifying the measure, it is with respect to any of those equivalent measures. While we do not have a general explicit formula for the density of~$\lambda^\epsilon_\textnormal{inv}$\,---\,decay and regularity are discussed in Appendix~\ref{app:cons}\,---, we have
\begin{equation}
	\label{eq:eq-meas}
	\mu_0^\epsilon(E) = \frac{\int_{E} \Exp{-\epsilon^{-1}V(x)} \d x}{\int_{\rr^N} \Exp{-\epsilon^{-1}V(y)} \d y}
\end{equation}
for all Borel subsets~$E$ of~$\rr^N$.

\subsection{Time reversal and the canonical entropy production functional}

Throughout the paper, we use the shorthand~$\cCt$ for the space~$C([0,t];\rr^N)$ of continuous paths in~$\rr^N$ over the time interval~$[0,t]$. It is always equipped with the supremum norm~$\|\,\cdot\,\|_\infty$; the corresponding Borel~$\sigma$-algebra is denoted~$\cB_t$.

We denote the distribution of~$(X^{x,\epsilon}_s)_{0 \leq s \leq t}$ in~\eqref{eq:SDE} by~$\cP_t^{x,\epsilon}$. This is a measure on~$(\cCt,\cB_t)$. With a slight abuse of notation,
we define~$\cP_t^{\lambda,\epsilon}$ as the analogous object but with random initial condition~$X_0^{\lambda,\epsilon} \sim \lambda$ (independent of~$W$) for a probability measure~$\lambda$ on~$\rr^N$. In other words,~$\cP_t^{\lambda,\epsilon}$ is the unique Borel measure on~$\cCt$ such that
\begin{equation}\label{eq:disint}
  \int_{\cCt} H(\gamma) \,\cP_t^{\lambda,\epsilon}(\d\gamma)  = \int_{\rr^N} \bigg( \int_{\cCt} H(\gamma)\, \cP_t^{x,\epsilon}(\d \gamma)\bigg) \, \lambda(\d x)
\end{equation}
for any nonnegative measurable function $H : \cCt \to \rr$.

The measures~$\mathcal{Q}_t^{x,\epsilon}$ and~$\mathcal{Q}_t^{\lambda,\epsilon}$ are defined analogously using~\eqref{eq:SDE-Y} {i.e.}\ the case~$b \equiv 0$. We have mentioned in the Introduction that~$\mathcal{Q}_t^{\lambda,\epsilon}$ is invariant under time reversal if and only if~$\lambda$ possesses a density proportional to~$\exp(-\epsilon^{-1}V)$. This is made more precise by the following identity:
\begin{equation}
\label{eq:change-ic}
	\log \frac{\d\mathcal{Q}_t^{\lambda,\epsilon}}{\d(\mathcal{Q}_t^{\lambda,\epsilon} \circ \Theta_t)}(\gamma) = \log \frac{\d\lambda}{\d\mu_0^\epsilon}(\gamma(0))
	-
	\log \frac{\d\lambda}{\d\mu_0^\epsilon}(\gamma(t))
\end{equation}
for $\mathcal{Q}_t^{\lambda,\epsilon}$-almost all~$\gamma \in \cCt$,
where $\pi_s : \cCt \to \rr^N$ is evaluation map $\gamma \mapsto \gamma(s)$, and where  \emph{time reversal} is the unique involution $\Theta_t : \cCt \to \cCt$ determined by the relation
$
	\pi_s \circ \Theta_t = \pi_{t-s}
	.
$
{We provide a proof of this identity in Appendix~\ref{app:cons}.}

The behaviour of~$\cP_t^{\lambda,\epsilon}$ under the time reversal~$\Theta_t$ is in general more subtle and, consistently with the intuition from thermodynamics, the dependence of the Radon--Nikodym derivative is not limited to the initial and final conditions of the path. The proof we give of the proposition below uses comparison with~$\mathcal{Q}_t^{\lambda,\epsilon}$, as in~e.g.~\cite[\S{5.4}]{JPS17}. Another possible route is to use the results of~\cite{HP86} on the reversal of~$(X_{s}^{\lambda,\epsilon})_{s \geq 0}$.

\begin{proposition}
\label{prop:RN}
	Under Assumptions~\textnormal{(L0)} and~\textnormal{(L1)}, if~$\lambda$ and the Lebesgue measure are mutually absolutely continuous, then~$\cP_t^{\lambda,\epsilon}$ and~$\cP_t^{\lambda,\epsilon} \circ \Theta^{-1}$ are mutually absolutely continuous and
	\begin{equation}
		\log \frac{\d\cP_t^{\lambda,\epsilon}}{\d(\cP_t^{\lambda,\epsilon} \circ \Theta_t^{-1})}(\gamma)
		=
		\log \frac{\d\lambda}{\d\mu_0^\epsilon}( \gamma(0))
		-
		\log \frac{\d\lambda}{\d\mu_0^\epsilon}( \gamma(t))
		+
		\frac{1}{\epsilon} \int \braket{b(\gamma),\circ\d\gamma}
	\label{eq:RN}
	\end{equation}
	for $\cP_t^{\lambda,\epsilon}$-almost all~$\gamma \in \cCt$, where~$\mu_0^\epsilon$ is defined by~\eqref{eq:eq-meas}.
  Moreover,
  \begin{equation}
  \label{eq:ES}
    \int_{\cCt} \bigg(\frac{\d\cP_t^{\lambda,\epsilon}}{\d(\cP_t^{\lambda,\epsilon} \circ \Theta_t^{-1})}(\gamma)\bigg)^\alpha \, \cP_t^{\lambda,\epsilon}(\d\gamma)
    = \int_{\cCt} \bigg(\frac{\d\cP_t^{\lambda,\epsilon}}{\d(\cP_t^{\lambda,\epsilon} \circ \Theta_t^{-1})}(\gamma)\bigg)^{1-\alpha} \, \cP_t^{\lambda,\epsilon}(\d\gamma)
  \end{equation}
  for all~$\alpha \in \rr$ for which both sides are finite. Both sides of~\eqref{eq:ES} are log-convex in~$\alpha$.
\end{proposition}

Before we proceed with the proof, let us briefly clarify the meaning of the expression~\eqref{eq:RN}. On the canonical probability space~$(\cCt, \mathcal{B}_t, \cP_t^{\lambda,\epsilon})$, the random variable $W = (W_s)_{s\geq 0}$ defined by
\begin{equation}\label{eq:BM-as-fn-of-path-}
  W_s(\gamma) := \frac{1}{\sqrt{2\epsilon}} \Big( \gamma(s) - \gamma(0) - \int_0^s \big(-\nabla V(\gamma(r)) + b(\gamma(r))\big) \d r \Big)
\end{equation}
is a Brownian motion, the evaluation map~$\pi_0$ has distribution~$\lambda$ and is independent of~$W$, and the canonical process~$(\pi_s)_{s \geq 0}$ is the unique solution to the~\textsc{sde}~\eqref{eq:SDE} with initial condition~$\lambda$. Hence,~$(\pi_s)_{s \geq 0}$ is a continuous semimartingale and we allow ourselves notational shortcuts such as
\[
\int \braket{b(\gamma), \circ \d\gamma} := \Big(\int_0^t \braket{b(\pi_s),\circ \d \pi_s}\Big)(\gamma)
\]
and
\[
\int \braket{b(\gamma),\d W(\gamma)} := \Big(\int_0^t \braket{b(\pi_s),\d W_s}\Big)(\gamma),
\]
where the right-hand sides are defined $\cP_t^{\lambda,\epsilon}$-almost surely according to the usual theory of Stratonovich and It\^o stochastic integration with respect to continuous semimartingales~\cite[\S{II.4--II.7}]{Pro}.

\begin{proof}[Proof of Proposition~\ref{prop:RN}]
  Throughout the proof, we omit keeping explicitly track of the dependence on~$\epsilon$ in the notation. We first reduce the general case to the technically easier case where the nonconservative vector field has compact support. For the latter, we suppose that the reader is familiar with Girsanov's theorem and related criteria; see e.g.~\cite[\S{III.8}]{Pro}

	Once mutual absolute continuity is proved, the symmetry expressed in~\eqref{eq:ES} is an immediate consequence of the definition of the Radon--Nikodym derivative and the fact that~$\Theta_t$ is an involution. Log-convexity is a consequence of H\"older's inequality.
  \begin{description}
    \item[Step 1: Reduction to the case where~$b$ has compact support.] Suppose that the proposition has been proved in the case where~$b$ has compact support.
    For~$R \in \nn$, pick a globally Lipschitz vector field~$b_{R}$ satisfying~$|b_{R}(x)| \leq |b(x)|$ for all~$|x| \in \rr^N$, $b_{R}(x) = b(x)$ whenever~$|x| \leq {R-1}$, and~$b_{R}(x) = 0$ whenever~$|x| \geq R$.

  	Let $\cP^{\lambda}_t[R]$ be the path measure associated to the~\textsc{sde} with initial condition~$\lambda$ and drift~$b_{R}$, and let~$B_{R}$ denote the centered open ball of radius~$R$ in~$\cCt$. Observe that~$B_R$ is invariant under~$\Theta_t$ and that
  	\[
  		\cP_t^{\lambda}[R](\Gamma \cap B_R) = \cP_t^{\lambda}(\Gamma \cap B_R)
  	\]
  	for all Borel sets~$\Gamma \subseteq \cCt$; see e.g.~the construction in~\cite[\S{3.4}]{Has}. Hence, by hypothesis,~$(\cP_t^{\lambda} \circ \Theta_t^{-1}) (B_R \cap \,\cdot\,)$ is absolutely continuous with respect to~$\cP^{\lambda}_t(B_R \cap \,\cdot\,)$ and
  	\begin{equation*}
  		F_R(\gamma)
  		:=
  		\one_{B_R}(\gamma) \frac{\d\lambda}{\d\mu_0}(\gamma(0))\frac{\d\mu_0}{\d\lambda}(\gamma(t))
  		\exp\Big(\frac{1}{\epsilon} \int \braket{b_R(\gamma),\circ\d\gamma}\Big)
  	\end{equation*}
  	is a Radon--Nikodym derivative.
  	Because,
  	\[
  		\lim_{R \to \infty} \cP_t^{\lambda}(B_R) = 1
  	\]
  	by Lemma~\ref{lem:dissipation-with-sup}, we can deduce that $\cP_t^{\lambda}  \circ \Theta_t$ is absolutely continuous with respect to~$\cP_t^{\lambda}$, with a Radon--Nikodym derivative
  	\[
  		F(\gamma) := \lim_{R \to \infty} F_R(\gamma).
  	\]

  	The fact that $F(\gamma)$ is strictly positive  and equals the right-hand side of~\eqref{eq:RN} follows from basic properties of the exponential, the absolute-continuity assumption and the fact that
  	\[
  		\lim_{R \to \infty}  \int \braket{b_R(\gamma),\circ\d\gamma} =  \int \braket{b(\gamma),\circ\d\gamma}
  	\]
  	for all~$\gamma \in \cCt$ (given~$\gamma$, take $R > \|\gamma\|_\infty$). \\

    \item[Step 2: Proof in the case where~$b$ has compact support]\hspace{-.5em}.
      \begin{description}
        \item[Step 2a: Comparing~$\cP_t^{\lambda,\epsilon}$ and~$\cQ_t^{\lambda,\epsilon}$.]
        Because~$b$ is bounded, Novikov's condition is satisfied and the process~$(Z_s)_{s \in [0,t]}$ defined by the Dol\'eans-Dade exponential
      	\begin{equation}
      		Z_s(\gamma) := \exp\Big(\frac{1}{\sqrt{2\epsilon}} \int_0^s \braket{b(\gamma),\d W(\gamma)} - \frac{1}{4 \epsilon}\int_0^s |b(X^{\lambda,\epsilon}_r(\gamma))|^2\d r \Big)
      	\label{eq:def-Z-mart}
      	\end{equation}
      	is a martingale. Hence, by Girsanov's theorem, the process~$(w_s)_{s \in [0,t]}$ defined by
      	\begin{equation}
      		w_s(\gamma) := W_s(\gamma) - \frac{1}{\sqrt{2\epsilon}}\int_0^{s} b(\gamma(r)) \d r
      		\label{eq:Girs-chg}
      	\end{equation}
      	is a Brownian motion with respect to the measure
      	$
      		\d P_t = Z_t \d \mathcal{Q}_t^{\lambda,\epsilon}
      	$.
      	Substituting~\eqref{eq:Girs-chg} into~\eqref{eq:SDE-Y} and comparing with~\eqref{eq:SDE}, we deduce that~$P_t = \cP_t^{\lambda,\epsilon}$, that is
        \begin{equation*}
          \frac{\d \cP_t^{\lambda,\epsilon}}{\d \mathcal{Q}_t^{\lambda,\epsilon}}(\gamma) = Z_t(\gamma) > 0.
        \end{equation*}
        \item[Step 2b: Comparing~$\cP_t^{\lambda,\epsilon}$ and~$\cP_t^{\lambda,\epsilon} \circ \Theta_t^{-1}$.] Combining~\eqref{eq:change-ic} and Step~2a, we have
        \begin{equation}
        \label{eq:RNd-almost-done}
          \frac{\d \cP_t^{\lambda,\epsilon}}{\d \cP_t^{\lambda,\epsilon} \circ \Theta_t^{-1}}(\gamma) =
          {Z_t(\gamma)} \frac{\d\lambda}{\d\mu_0^\epsilon}(\pi_0\gamma)\frac{\d\mu_0^\epsilon}{\d\lambda}(\pi_t\gamma) \frac{1}{Z_t(\Theta_t\gamma)}.
        \end{equation}
        Using the identity
        \begin{align*}
        	W_s(\Theta_t \gamma) 	&=  -(W_{t-s}(\gamma) - W_t(\gamma)) + \frac{2(\pi_s\Theta_t\gamma - \pi_t\gamma)}{\sqrt{2\epsilon}}
        \end{align*}
        following from~\eqref{eq:BM-as-fn-of-path-} in a sequence of approximations of $Z_t(\Theta_t \gamma)$ by discretisation of~$(b \circ \pi_s)_{s \in [0,t]}$ using a random partition of~$[0,t]$ tending to the identity in the sense of~\cite[\S{II.5}]{Pro}, we find
        \begin{align*}
      		\int \braket{ b(\Theta_t\gamma),\d W(\Theta_t\gamma)} &= \int \braket{b(\gamma),\d W(\gamma)}
            - \frac{2}{\sqrt{2\epsilon}} \int \braket{b(\gamma), \circ \d \gamma}.
      	\end{align*}
        Therefore,
        \[
          \log \frac{Z_t(\gamma)}{Z_t(\Theta_t\gamma)} = \frac{1}{\epsilon} \int \braket{b(\gamma), \circ \d \gamma}
        \]
        and taking the logarithm of~\eqref{eq:RNd-almost-done} gives the proposed formula. \qedhere
      \end{description}
  \end{description}
\end{proof}

The logarithm of the Radon--Nikodym derivative in Proposition~\ref{prop:RN} is called the \emph{canonical entropy production functional} in~\cite{JPS17}. We note the following immediate corollary of Proposition~\ref{prop:RN}, the explicit formula~\eqref{eq:eq-meas} for~$\mu_0^\epsilon$ and well-known properties of Stratonovich integrals; see e.g.~\cite[\S{V.5}]{Pro}.

\begin{corollary}
  Under Assumption~\textnormal{(L0)} and~\textnormal{(L1)}, if~$\lambda$ and the Lebesgue measure are mutually absolutely continuous, then~$\cP_t^{\lambda,\epsilon}$ and~$\cP_t^{\lambda,\epsilon} \circ \Theta_t^{-1}$ are mutually absolutely continuous and
  \begin{equation*}
		\log \frac{\d\cP_t^{\lambda,\epsilon}}{\d(\cP_t^{\lambda,\epsilon} \circ \Theta_t^{-1})}(\gamma)
		=
		\log \frac{\d\lambda}{\d\vol}(\gamma(0)) - \log \frac{\d\lambda}{\d\vol}(\gamma(t))
		+ \frac{1}{\epsilon} \int \braket{-\nabla V(\gamma) + b(\gamma),\circ\d\gamma}
	\end{equation*}
  for $\cP_t^{\lambda,\epsilon}$-almost all~$\gamma \in \cCt$.
\end{corollary}

\section{Generalised entropy production functionals}
\label{sec:EPF}

Motivated by the structure revealed in the previous section and by~\cite{Ku98,LS99,CVZ03,MNV03,JPS17}, we introduce a family of entropy production functionals parametrised by the choice of boundary terms. Throughout the remainder of the paper, the function~$g : \rr^N \to (0,\infty)$ is continuous, and the initial condition~$\lambda$ is absolutely continuous with respect to the Lebesgue measure and has finite second moment.

\subsection{Definition and the weak law of large numbers}

\begin{definition}
	The \emph{entropy production functional} associated to the function~$g$ is the function~$\ss_t^\epsilon$ defined by
	\begin{equation}
	\label{eq:def-EPF}
		\ss_t^\epsilon(\gamma) := \log g(\gamma(0)) - \log g(\gamma(t)) + \frac{1}{\epsilon} \int \braket{b(\gamma),\circ\d\gamma},
	\end{equation}
	considered as a random variable on~$\cCt$ with respect to the probability measure~$\cP_t^{\lambda,\epsilon}$. For~$\alpha \in \rr$, we use
	\[
		\chi_t^\epsilon(\alpha) :=
		\int_{\cCt} \Exp{-\alpha \ss_t^\epsilon }\d\cP^{\lambda,\epsilon}_t
	\]
	for the \textsc{mgf} of~$\ss_t^\epsilon$ in~$\alpha$. We speak of a \emph{steady-state} functional if the initial condition~$\lambda$ entering the definition of~$\cP^{\lambda,\epsilon}_t$ equals the invariant measure~$\lambda_{\textnormal{inv}}^\epsilon$.
\end{definition}

\begin{remark}
  The choice of~$-\alpha$ in the exponent is common in the physics literature and is made here to facilitate the identification of certain symmetries. Indeed, the symmetry noted in Proposition~\ref{prop:RN} can be used to deduce $\chi_t^\epsilon(1-\alpha) = \chi_t^\epsilon(\alpha)$ in the case~$g = \d\lambda/\d\mu_0^\epsilon$.
  However, this symmetry at finite~$t$ is not expected to hold for a generic choice of~$g$. This choice also has an incidence on our choice of sign for the Legendre transform in Sections~\ref{sec:large-t} to~\ref{sec:small-e}.
\end{remark}

We are mainly interested in the large deviations of~$\ss_t^\epsilon$ as~$t \to \infty$ and then~$\epsilon \to 0$. To tackle this problem, we will need additional assumptions on the behaviour of the boundary term~$g$ at infinity. Before we do so, let us state and prove a weak law of large numbers which holds under minimal assumptions on the decay of the boundary terms.

\begin{proposition}\label{prop:lln-t}
	Suppose that Assumptions~\textnormal{(L0)} and~\textnormal{(L1)} are satisfied and let
	\begin{equation}\label{eq:moy}
		\mathfrak{m}^\epsilon := \int_{\rr^N} \big(\epsilon^{-1}|b|^2 - \epsilon^{-1}\braket{b,\nabla V}  \big) \d\lambda_\textnormal{inv}^\epsilon + \int_{\rr^N}\operatorname{div} b \, \d\lambda_\textnormal{inv}^\epsilon.
	\end{equation}
	Then, for all~$\delta > 0$,
	\begin{equation}\label{eq:lln}
		\lim_{t \to \infty} \cP_t^{\lambda,\epsilon} \Big\{ \Big| \frac 1t \ss^\epsilon_t - \mathfrak{m}^\epsilon \Big| > \delta \Big\} = 0.
	\end{equation}
\end{proposition}

\newcommand{\noepsilon}{\phantom{.}}
\newcommand{\cnoepsilon}{\phantom{.}}

\begin{proof}
  We prove the case~$\epsilon = 1$ and leave it to the reader to check the presence of the appropriate factors of~$\epsilon$ in~\eqref{eq:moy}. The Stratonovich integral in the definition~\eqref{eq:def-EPF} of~$t^{-1}\ss^{\noepsilon}_t$ can be decomposed as
  \begin{equation*}
	\begin{split}
     \frac{1}{t}\int_0^t \braket{b(\gamma),\circ\d\gamma}  = \frac{1}{t}\int_0^t \big(|b(\pi_s\gamma)|^2 - \braket{b(\pi_s \gamma), \nabla V(\pi_s\gamma)}+ \noepsilon \operatorname{div}  b(\pi_s\gamma)\big) \d s \qquad
		 \\
		 	+ \frac{\sqrt{2\noepsilon}}{t} \int_0^t \braket{b(\gamma),\d W(\gamma)}.
	\end{split}
  \end{equation*}
  The integral on the first line of the right-hand is admissible for an application of the law of large numbers for continuous functions of~$X_s^{\noepsilon}$\,---\,see e.g.\ Theorem 4.2 in~\cite[Ch.\,4]{Has}\,---, which yields
  \begin{equation*}
    \lim_{t\to\infty} \cP_t^{\lambda} \Big\{ \Big| \frac{1}{t}\int_0^t \big(|b(\pi_s\gamma)|^2 - \braket{b(\pi_s \gamma), \nabla V(\pi_s\gamma)} + \operatorname{div} b(\pi_s\gamma)\big) \d s
  - {\mathfrak{m}} \Big| > \tfrac 15 \delta \Big\} = 0.
	\end{equation*}
  The integral on the second line of the right-hand side is a martingale. For integer times, the hypotheses of the law of large numbers for discrete-time martingales in~\cite[\S{VII.8}]{Fel2} are satisfied thanks to It\^o's isometry and Lemma~\ref{lem:dissipation-without-sup}.
	Hence,
  \[
    \lim_{t\to\infty} \cP_t^{\lambda\cnoepsilon} \Big\{ \Big|\frac{\sqrt{2\noepsilon}}{\lfloor t \rfloor} \int_0^{\lfloor t \rfloor} \braket{b(\gamma),\d W(\gamma)}\Big| > \tfrac 15 \delta \Big\} = 0.
  \]
  By Chebysh\"ev's inequality and It\^o's isometry, we have
  \begin{align*}
    \cP_t^{\lambda\cnoepsilon} \Big\{ \Big| \frac{\sqrt{2\noepsilon}}{\lfloor t \rfloor} \int_{\lfloor t \rfloor}^t \braket{b(\gamma),\d W(\gamma)}\Big| > \tfrac 1 5 \delta \Big\}
    &\leq \frac{50 \noepsilon }{\delta^2 \lfloor t \rfloor^2 } \int_{\cCt} \int_{\lfloor t \rfloor}^t |b(\pi_s\gamma)|^2 \d s \,\cP_t^{\lambda\cnoepsilon}(\d\gamma),
  \end{align*}
  with the double integral on the right-hand side bounded uniformly in~$t$ by Tonelli's theorem, Lemma~\ref{lem:dissipation-without-sup} and the fact that~$b$ is globally Lipschitz.

  As for the boundary terms in the definition of~$t^{-1}\ss^{\noepsilon}_t$, we note that it is no loss of generality to assume that $g(0) = 1$. Then, by positivity and continuity of~$g$, there exists a monotone family~$(R_{M})_{M > 0}$ of radii properly diverging to~$+ \infty$ with~$M$ such that
	\begin{equation*}
		g^{-1}\big([\Exp{-M/5},\Exp{M/5}]\big) \supseteq \{x \in \rr^N : |x| \leq R_M\}.
	\end{equation*}
	Using this inclusion with $M=t\delta$,
  \begin{align*}
    \cP_t^{\lambda\cnoepsilon}\{|t^{-1} \log g (\pi_0\gamma)| > \tfrac 1 5 \delta \}
		&\leq \cP_t^{\lambda\cnoepsilon}\{|\pi_0\gamma| \geq R_{t\delta} \}
    = \lambda \{x \in \rr^N : |x| \geq  R_{t\delta} \}
  \end{align*}
  converges to~$0$ as~$t \to \infty$ because~$R_{t\delta} \to \infty$ and the initial condition~$\lambda$ is a probability measure. Using the same inclusion, Chebysh\"ev's inequality and Lemma~\ref{lem:dissipation-without-sup},
  \begin{align*}
    \cP_t^{\lambda\cnoepsilon}\{|t^{-1} \log g (\pi_t\gamma)| > \tfrac 1 5 \delta \}
    &\leq \cP_t^{\lambda\cnoepsilon}\{ |\pi_t\gamma| \geq R_{\delta t} \}
		\leq \frac{\int \braket{y,H_b y} \, \lambda(\d y) + C }{R_{\delta t}^2 \inf \sp H_b}
  \end{align*}
  also converges to~$0$ as~$t \to \infty$. The proof is then concluded using the triangle inequality and a union bound.
\end{proof}

{At this stage, one can already use 2-dimensional examples with linear~$b$, quadratic~$V$ and explicit Gaussian~$\lambda_\textnormal{inv}^\epsilon$ to exhibit cases where $\mathfrak{m}^\epsilon > 0$, independently of~$\epsilon$, a strict inequality we consider as a key feature of nonequilibrium phenomena.
This observation\,---\,which we push in Proposition~\ref{prop:locally-a-grad} below\,---\,complements L.~Bertini and G.~Di Ges\`u's discussion of the typical behaviour and of the order of the limits $\epsilon \to 0$ and $t \to \infty$ for their \emph{rescaled} functional in~\cite[\S{2}]{BD15}.
On the point of rescaling, let us also mention that the \textsc{ldp} in~\cite[\S{4}]{BD15} makes rigorous the existence of a rate function~$I$ describing, in our notation, the asymptotics
\[
	\pp\{t^{-1} S_t^{\textnormal{W}, \epsilon} \approx \epsilon^{-1}\varsigma\} \asymp \exp\left(-t\epsilon^{-1} I(\varsigma)\right)
\]
for $1 \ll t \ll \epsilon^{-1}$ and all $\varsigma \in \rr$. Note the difference in scaling when compared with~\eqref{eq:heur-ldp-van}. This rate function~$I$ always vanishes at~$\varsigma = 0$.
To see this, consider, in Section~2 there and with their notation, the path~$\varphi$ which is constantly at a common stationary point~$x_*$ of~$b$ and~$-\nabla V + b$ so that $\varphi \in \mathcal{A}_T^{x_*}(0)$ and $I^{x_*}_T(\varphi) = 0$; this implies~$S^{x_*x_*}_T(0) = 0$ and the limiting rate function in Eq.~(4.1) of Theorem~4.1 must vanish at~$0$. On the other hand, the above discussion of the law of large numbers suggests\,---\,and this will also be confirmed in Proposition~\ref{prop:locally-a-grad}\,---\,that the rate function we are about to obtain may very well remain strictly positive at~$s = 0$ as~$\epsilon \to 0$.
}

We also see from the formula~\eqref{eq:moy} that the behaviour as~$\epsilon \to 0$ of the mean entropy production per unit time~$\mathfrak{m}^\epsilon$ will depend on that of~$\lambda_\textnormal{inv}^\epsilon$ and hence on the Freidlin--Wentzell quasipotential~\cite[\S{6--8}]{FW70} associated to the ordinary differential equation, $\dot X = -\nabla V(X) + b(X)$.
In situations where the quasipotential is proportional to~$V$, more detailed information can be obtained and points where~$V$ attains its global minimum play a particular role. We will come back to this in Section~\ref{sec:small-e}.

\subsection{Assumptions on the boundary terms and the initial condition}
\label{ssec:b-term-assumptions}

As mentioned in the Introduction, the \textsc{ldp} at the heart of this article is local. At the technical level, this is due to the fact that we are able to prove convergence of the rescaled logarithm of the~\textsc{mgf},~$t^{-1}\log\chi_t^\epsilon(\alpha)$, as $t \to \infty$ and then as $\epsilon \to 0$ only for certain values of~$\alpha$.

In the special case where~$b$ is bounded and orthogonal to~$\nabla V$, and~$g \equiv 1$, L.~Bertini and G.~Di Ges\`u have shown convergence as~$t \to \infty$ for all~$\alpha \in \rr$, without the type of assumption we are about to introduce~\cite[App.\,A]{BD15}. However, the analysis of the linear case in~\cite{JPS17} shows the intricacies of taking the limit~$t \to \infty$ for~$\alpha \notin [0,1]$ in the case where~$b$ is unbounded, as well as the sensitivity of the limit to the choice of boundary terms.
Subsequently taking the limit~$\epsilon \to 0$ for~$\alpha$ outside~$[0,1]$ also comes with its own complications{; see Remark~\ref{rem:is-it-global} below.}
We restrict our attention to~$\alpha$ in the interval
\begin{equation}
\label{eq:cA-def}
	\cA := \bigcup_{\ell \in (0,1)} \interior\{\alpha : \ell \tfrac 14 |\nabla V(x)|^2  - \tfrac 12\braket{b,\nabla V(x)} + \alpha(1-\alpha) |b(x)|^2 \geq 0  \text{ for all } x\in\rr^N\}.
\end{equation}
One can use~\textnormal{(RB)} to show that if~$\ell$ is close enough to~$1$, then the quantity of interest is nonnegative for~$\alpha$ in an open interval containing~$[0,1]$; see Lemma~\ref{lem:link-I-Abb}. The interval~$\cA$ is symmetric about~$\alpha = \tfrac 12$. Another possible obstruction is the behaviour of the boundary term~$g$ used in the construction of~$\ss_t^\epsilon$. We introduce the following technical assumption and immediately give more tractable sufficient conditions.

\begin{description}
	\item[Assumption (IP).] There exists an open interval~$I^\epsilon$ with $[0,1] \subset I^\epsilon \subseteq \cA$ and such that the following property holds for all~$\alpha \in I^\epsilon$: there exists $p_\alpha^\epsilon \in (1,\infty)$ and $\ell_{\alpha}^{\epsilon} \in (0,1)$ such that
		\begin{equation}\label{eq:admiss-I}
	 		\ell_{\alpha}^{\epsilon} \tfrac {1}{{p_\alpha^\epsilon}} \big(1 - \tfrac 1{p_\alpha^\epsilon} \big) |\nabla V(x)|^2  - \tfrac{1 -2\alpha +\alpha {p_\alpha^\epsilon} }{{p_\alpha^\epsilon}}\braket{b(x),\nabla V(x)} + \alpha(1-\alpha) |b(x)|^2 \geq 0,
	  \end{equation}
		for all~$x \in \rr^N$, and both
			\begin{equation}\label{eq:I}
				g^\alpha \in \mathrm{L}^{p_\alpha^\epsilon}(\rr^N,\d\mu_0^\epsilon) \quad \text{and} \quad \tfrac{\d\lambda}{\d\mu_0^\epsilon} \g^{-\alpha} \in \mathrm{L}^{q_\alpha^\epsilon}(\rr^N,\d\mu_0^\epsilon),
			\end{equation}
		 with $\tfrac{1}{p_\alpha} + \tfrac{1}{q_\alpha} = 1$.
\end{description}

In the case $g \equiv 1$, we can give a simple condition on the initial condition~$\lambda$ which is sufficient for Assumption~\textnormal{(IP)} to hold.
The proof elucidates why we leave~$p_\alpha^\epsilon$ as a parameter instead of fixing $p_\alpha^\epsilon = q_\alpha^\epsilon = 2$: it is this parameter which allows us to accommodate measures~$\lambda$ for which $\tfrac{\d\lambda}{\d\mu_0^\epsilon} \in \mathrm{L}^{2-\delta}(\rr^N,\d\mu_0^\epsilon)$ with~$\delta > 0$ small, but not with~$\delta = 0$; {cf.}\ Lemma~\ref{lem:ub-ic-RN}.

\begin{lemma}
	Suppose that Assumptions~\textnormal{(L0)},~\textnormal{(L1)} and~\textnormal{(RB)} are satisfied and that~$g \equiv 1$. With~$k_b \in [0,\tfrac 12)$ as in~\textnormal{(RB)}, if there exists~$\delta \in (0,\tfrac 12)$ such that the initial condition~$\lambda$ satisfies
	\begin{equation}\label{eq:I-suff-1}
		\frac{\d\lambda}{\d\mu_0^\epsilon} \in \mathrm{L}^{\frac{1}{1 - k_b - \delta}}(\rr^N,\d\mu_0^\epsilon),
	\end{equation}
	then Assumption~\textnormal{(IP)} is satisfied.
\end{lemma}

\begin{proof}
	For any given~$p \in (\tfrac{1}{1-k_b},\tfrac{1}{k_b})$, there exists~$\ell \in (0,1)$ such that condition~\eqref{eq:admiss-I} holds for all~$\alpha$ in an open interval containing~$[0,1]$; see Lemma~\ref{lem:link-I-Abb}.
	Without loss of generality, $\delta > 0$ in~\eqref{eq:I-suff-1} is small enough that $p = \tfrac{1}{k_b+\delta}$ is such a value, but then
	we have $q = \tfrac{1}{1-k_b-\delta}$ and~\eqref{eq:I} with~$g \equiv 1$ reduces to~\eqref{eq:I-suff-1}.
\end{proof}

In the steady-state canonical case, $g = \d\lambda/\d\mu_0^\epsilon$ and $\lambda = \lambda_\textnormal{inv}^\epsilon$,~\eqref{eq:I} is guaranteed to hold for all~$\alpha \in [0,1]$. Indeed, one can apply Lemma~\ref{lem:ub-ic-RN} with some~$p_\alpha^\epsilon$ close enough to 2 that~\eqref{eq:admiss-I} holds.
However, obtaining~\eqref{eq:I} outside the interval~$[0,1]$ is in general a delicate task which, to our knowledge, requires extra technical assumptions{\,---\,unless $\braket{b,\nabla V} = \epsilon \operatorname{div} b$, in which case $ g = \d\lambda_\textnormal{inv}^\epsilon/\d\mu_0^\epsilon \equiv 1$ and~\eqref{eq:I} trivially holds}.

\begin{lemma}
	Suppose that Assumptions~\textnormal{(L1)} and~\textnormal{(RB)} are satisfied, that~$g = \d\lambda/\d\mu_0^\epsilon$ and that the initial condition is~$\lambda = \lambda_\textnormal{inv}^\epsilon$. If there exists $c_-, c_+, \gamma_-, \gamma_+ > 0$ and $a \geq 2$ such that
  \[
     V(x) \geq \gamma_- |x|^{a} - c_-
  \]
  and
  \[
    |\nabla V(x)|^2 \leq \gamma_{+}^2 a^2|x|^{2(a-1)} + c_+
  \]
  for all~$x \in \rr^N$, then Assumption~\textnormal{(IP)} is satisfied, uniformly in~$\epsilon$.
\end{lemma}

\begin{proof}
	We will show that there exists a nonempty interval of the form~$(\alpha_-,0]$ which does not depend on~$\epsilon$ and such that~\eqref{eq:admiss-I} and~\eqref{eq:I} hold for all~$\alpha$ in this interval, with common~$p$ and~$\ell$. A similar argument can be given to find an interval of the form~$[1,\alpha_+)$.

	Fix $p = 2 + \delta$ for some~$\delta > 0$ small enough that there exists~$\ell \in (0,1)$ such that condition~\eqref{eq:admiss-I} holds for all~$\alpha$ in a nonempty interval of the form~$(\hat{\alpha}_-,0]$.
	Then, $q \in (1,2)$ and the second inclusion in~\eqref{eq:I} for all~$\alpha$ in a nonempty interval of the form~$(\tilde{\alpha}_-,0]$ is guaranteed by Lemma~\ref{lem:ub-ic-RN}. Finally, we claim that the fact that
	\begin{equation}\label{eq:claim-conseq-ABG}
		\Big(\frac{\d\lambda_\textnormal{inv}^\epsilon}{\d\mu_0^\epsilon}\Big)^{\alpha} \in \mathrm{L}^p(\rr^N,\d\mu_0^\epsilon)
	\end{equation}
	for all~$\alpha$ in a nonempty interval of the form~$(\bar\alpha_-,0]$ follows from the work \cite[\S{4}]{ABG19}. We then take $\alpha_- := \max\{\hat{\alpha}_-,\tilde{\alpha}_-, \bar\alpha_-\}$ to complete the proof.

	To establish~\eqref{eq:claim-conseq-ABG}, pick $M > 0$ such that $|\braket{b(x), x}| \leq M |x|^2$ for all~$x \in \rr^N$.\footnote{We assume without loss of generality that our space coordinates are centered in such a way that $x=0$ is one of the critical points of~$V$ and thus a stationary point of~$b$ by (RB), and then use the fact that $b$ is globally Lipschitz.}
  Combining the upper bound with~$\gamma_+$ and Assumption~\textnormal{(RB)}, we obtain a constant~$\tilde{c}_+$ such that
  \[
    - \tfrac 1{4\epsilon}|\nabla V(x)|^2 + \tfrac{1}{2\epsilon}\braket{b(x),\nabla V(x)} - \tfrac 12 \upDelta V(x) - \operatorname{div} b(x) \geq -  \tfrac 1{8\epsilon} \gamma_+ (1-2k_b) |x|^{2(a-1)}
  \]
  if $|x|$ is sufficiently large.
  Set
  \[
    K := \frac 1{2a} \big( M + \sqrt{{M^2} +  \gamma_+ (1-2k_b) }\big).
  \]
  By Theorem~4.1 in~\cite{ABG19} applied to the conjugated Fokker--Planck operator
	\[
		\epsilon\Delta -\braket{b,\nabla} - \tfrac 1{4\epsilon}|\nabla V|^2 + \tfrac{1}{2\epsilon}\braket{b,\nabla V} - \tfrac 12 \upDelta V - \operatorname{div} b,
	\]
	there exist constants~$C_\epsilon > 0$ and~$r_\epsilon > 0$ such that the unique function~$\varphi^\epsilon$ such that
	\[
    \lambda_\textnormal{inv}^\epsilon (\d x) = \Exp{-(2\epsilon)^{-1} V(x)} \varphi^{\epsilon}(x) \d x
  \]
	satisfies
  \[
    \varphi^\epsilon(x) \geq C_\epsilon \Exp{-\epsilon^{-1} K |x|^a}
  \]
  whenever~$|x| > r_\epsilon$; also see Lemma~\ref{lem:pre-ub-ic-RN}. 
	This last inequality can be rewritten as
  \begin{align*}
    \frac{\d\lambda_\textnormal{inv}^\epsilon}{\d\mu_0^\epsilon}(x) &\geq \tilde{C}_\epsilon \Exp{-(2\epsilon)^{-1}V(x)} \Exp{-\epsilon^{-1} K |x|^a}
  \end{align*}
  with some~$\tilde{C}_\epsilon > 0$.
  Using the lower bound in~$\gamma_-$, there exists~$\tilde{r}_\epsilon > 0$ such that
  \[
    \Big(\frac{\d\lambda_\textnormal{inv}^\epsilon}{\d\mu_0^\epsilon}(x)\Big)^{-|\beta|} \leq \tilde{C}_\epsilon^{-1} \exp \bigg( \frac{|\beta|}{\epsilon}\bigg(\frac 12 + \frac{K}{2\gamma_-}\bigg) V(x)\bigg)
  \]
  whenever~$|x| > \tilde{r}_\epsilon$. We conclude that the claim~\eqref{eq:claim-conseq-ABG} indeed holds for all~$\alpha \in (\bar\alpha_-,0]$ with
	$
		\bar\alpha_-^{-1} := -p(\tfrac 12 + \tfrac{K}{2\gamma_-})
	$.
\end{proof}

\begin{remark}
\label{rem:is-it-global}
	{The care in choosing the constraints on~$\alpha$ and~$p_\alpha^\epsilon$ here is taken for two reasons: the unboundedness of~$b$
	and the desire to obtain detailed information on the rate function as~$\epsilon \to 0$.}
	{If one is interested in the limit~$t\to\infty$ only, then one may replace the nonnegativity conditions in~\eqref{eq:cA-def} and~\eqref{eq:admiss-I} with the existence of a finite (negative) lower bound, as considered in Appendix~\ref{app:semig}. In particular, if~$b$ is bounded and globally Lipschitz, then the derived \textsc{ldp} as $t\to\infty$ is global as long as the condition~\eqref{eq:I} on the boundary term holds for all~$\alpha$, as in~\textnormal{\cite[App.\,A]{BD15}}.}
	{However, even for a smooth and compactly supported~$b$, nonnegativity in~\eqref{eq:cA-def} and~\eqref{eq:admiss-I} is used crucially for the limit~$\epsilon \to 0$ through the properties of the auxiliary potential~$W_0$ introduced in~\eqref{eq:aux-pot} of Section~\ref{app:semiclass}, and the techniques used there cannot possibly yield a global result as is.
	Indeed, values of~$\alpha$ that are far from the interval~$[0,1]$ lead to changes to the structure of the minima of~$W_0$ that render maladapted the harmonic-like approximations used in Section~\ref{app:semiclass}.}
\end{remark}

\subsection{A representation for the moment-generating function}
\label{ssec:rep-mgf}

Under Assumption~\textnormal{(IP)}, we prove the validity of a commonly used representation of the \textsc{mgf}~$\chi_t^\epsilon(\alpha)$ in terms of a semigroup of operators acting on the space~$\mathrm{L}^{p_\alpha^\epsilon}(\rr^N,\d\mu_0^\epsilon)$ obtained by deformation of the infinitesimal generator
$
	\Lambda^{\epsilon,0} = \epsilon \upDelta  +  \braket{-\nabla V + b,\nabla }
$
of the semigroup associated to the \textsc{sde}~\eqref{eq:SDE}. The proof relies on preliminary results on elliptic operators collected in Appendix~\ref{app:semig}, based on~\cite{AGG+,Lan,MPRS05}.

\begin{proposition}\label{prop:mgf}
	Suppose that Assumptions~\textnormal{(L0)},~\textnormal{(L1)},~\textnormal{(RB)} and~\textnormal{(IP)} are satisfied. Then, for all~$\alpha \in I^\epsilon$, the~\textsc{mgf} $\chi_t^\epsilon(\alpha)$ is finite  and can be represented as
	\begin{equation}\label{eq:mgf}
		\chi_t^{\epsilon}(\alpha) = \int_{\rr^N} g^{-\alpha} \big(\Exp{t\Lambda^{\alpha,\epsilon}} g^{\alpha}\big)\d\lambda,
	\end{equation}
	where~$\Lambda^{\epsilon,\alpha}$ is the infinitesimal generator of a semigroup on $\mathrm{L}^{p_\alpha^\epsilon}(\rr^N,\d\mu_0^\epsilon)$ given by
	\begin{equation}\label{eq:def-Lambda}
		\Lambda^{\epsilon,\alpha}f = \epsilon \upDelta f +  \braket{-\nabla V + (1 - 2\alpha)b,\nabla f} - \tfrac{\alpha(1-\alpha)}{\epsilon} |b|^2 f + \tfrac {\alpha}{\epsilon} \braket{b,\nabla V} f - \alpha f \operatorname{div} b,
	\end{equation}
	for all $f \in C^2_\textnormal{c}(\rr^N)$.
\end{proposition}

\begin{proof}
  We use an approximation strategy similar to that in the proof of Proposition~\ref{prop:RN} and again omit keeping explicit track of~$\epsilon$.
  \begin{description}
      \item[Step 1: Reduction to the case where~$b$ has compact support.]
      Suppose that the proposition has been proved in the case where~$b$ has compact support. For a general $b$, let $(b_R)_{R \in \nn}$ be a sequence of compactly supported approximations of~$b$ as in the proof of Proposition~\ref{prop:RN}. Fix~$\alpha$  and set
    	\[
    		\Lambda^{\alpha}_R := \epsilon \upDelta +  \braket{-\nabla V + (1 - 2\alpha)b_R,\nabla } - \tfrac{\alpha(1-\alpha)}{\epsilon} |b_R|^2  + \tfrac {\alpha}{\epsilon} \braket{b_R,\nabla V} - \alpha\operatorname{div} b_R
    	\]
    	and
    	\[
    	\Lambda^{\alpha} := \epsilon \upDelta +  \braket{-\nabla V + (1 - 2\alpha)b,\nabla } - \tfrac{\alpha(1-\alpha)}{\epsilon} |b|^2  + \tfrac {\alpha}{\epsilon} \braket{b,\nabla V} - \alpha\operatorname{div} b.
    	\]
    	It is shown in Appendix~\ref{app:semig} that these operators, considered with $\mathrm{W}^{2,{p_\alpha}}(\rr^N;\d\mu_0)$ as their domain, generate semigroups on~$\mathrm{L}^{p_\alpha}(\rr^N;\d\mu_0)$ if $\alpha \in \cA$.
			Hence, in view of~\eqref{eq:I} and the definition of~$I$, the right-hand side of~\eqref{eq:mgf} is well defined and finite for all~$\alpha \in I$. Therefore,  by~\eqref{eq:disint}, it suffices to show that
    	\begin{align}\label{eq:suff-for-prop:mgf}
    		 \big(\Exp{t \Lambda^\alpha } g^{\alpha} \big)(x)
    			&=  g^{\alpha}(x) \int_{\cCt} \Exp{-\alpha \ss_t }\d\cP^{x}_t
    	\end{align}
    	for almost all~$x \in \rr^N$.

    	One can show using the isometry between $\mathrm{L}^{p_\alpha}(\rr^N,\d\mu_0)$ and $\mathrm{L}^{p_\alpha}(\rr^N,\d\vol)$ and the second resolvent identity that
    	\[
    	 \operatorname*{s.r.-lim}_{R \to \infty} \Lambda^\alpha_{R} = \Lambda^\alpha.
    	\]
    	Hence, by Theorem~2.16 in~\cite[Ch.\,{IX}]{Kat},
    	\[
    		\operatorname*{s.-lim}_{R \to \infty} \Exp{t \Lambda^\alpha_R} = \Exp{t \Lambda^\alpha}
    	\]
    	on~$\mathrm{L}^{p_\alpha}(\rr^N,\d\mu_0)$, which contains~$g^\alpha$ by~\eqref{eq:I} of Assumption~\textnormal{(IP)}. In particular, there exists a subsequence~$(R_k)_{k \in \nn}$ properly diverging to~$+\infty$ such that
    	\[
    		\big(\Exp{t \Lambda^\alpha} g^{\alpha} \big)(x) = \lim_{k \to \infty} \big(\Exp{t\Lambda^\alpha_{R_k}} g^{\alpha} \big)(x)
    	\]
    	for almost all~$x \in \rr^N$. Hence, by hypothesis,
    	\begin{align*}
    		\big(\Exp{t \Lambda^\alpha } g^{\alpha} \big)(x)
    			&= \lim_{k\to\infty} \big(\Exp{t \Lambda^\alpha_{R_k} } g^{\alpha}  \big)(x) \\
    			&= \lim_{k\to\infty}  g^{\alpha}(x)\int_{\cCt} \Exp{-\alpha \ss_t }\d\cP^{x}_t[R_k] \\
    			&=   g^{\alpha}(x) \lim_{k\to\infty} \Big(\int_{B_{R_k}} \Exp{-\alpha \ss_t }\d\cP^{x}_t[R_k] +  \int_{B_{R_k}^\mathsf{C}} \Exp{-\alpha \ss_t }\d\cP^{x}_t[R_k] \Big),
    	\end{align*}
    	where $\cP^{x}_t[R_k]$ is the measure on the paths associated to the~\textsc{sde} with initial condition~$x$ and drift~$-\nabla V + b_{R_k}$, and where~$B_{R_k}$ denotes the ball of radius~$R_k$ in~$\cCt$. Using uniqueness,
    	\begin{align}
    		 \big(\Exp{t \Lambda^\alpha} g^{\alpha}  \big)(x)
    			&=  g^{\alpha}(x) \lim_{k\to\infty} \Big(\int_{B_{R_k}} \Exp{-\alpha \ss_t }\d\cP^{x}_t +  \int_{B_{R_k}^\mathsf{C}} \Exp{-\alpha \ss_t }\d\cP^{x}_t[R_k] \Big).
    	\label{eq:lim-k-ball-compl}
    	\end{align}
    	By Lebesgue monotone convergence,
    	\begin{align}\label{eq:LMC-in-R}
    		\lim_{k\to\infty} \int_{B_{R_k}} \Exp{-\alpha \ss_t }\d\cP^{x}_t
    			&= \int_{\cCt} \Exp{-\alpha \ss_t }\d\cP^{x}_t.
    	\end{align}
    	Note that this limit must be finite because the left-hand side of~\eqref{eq:lim-k-ball-compl} is finite, $g^{\alpha}(x)$ is strictly positive and the integral over the complement of the ball on the right-hand side of~\eqref{eq:lim-k-ball-compl} is nonnegative. Because~$I$ is open,
    	\[
    		\int_{\cCt} \Exp{-(\alpha+\delta\alpha) \ss_t }\d\cP^{x}_t < \infty
    	\]
    	as well if~$\delta > 0$ is small enough. Hence, we may apply H\"older's inequality with exponents~$1 + \delta$ and $(1-(1+\delta)^{-1})^{-1}$ to derive
    	\begin{align*}
    		\lim_{k\to\infty} \Big|\int_{B_{R_k}^{\mathsf{C}}} \Exp{-\alpha \ss_t }\d\cP^{x}_t\Big|
         \leq \lim_{k\to\infty} \Big|\int_{\cCt} \Exp{-(\alpha + \delta \alpha) \ss_t }\d\cP^{x}_t\Big|^{\frac 1{1+\delta}} \big(1 -\cP^{x}_t[{R_k}](B_{R_k})\big)^{1-(1+\delta)^{-1}},
    	\end{align*}
			which is controlled by Lemma~\ref{lem:dissipation-with-sup}. Using this bound and~\eqref{eq:LMC-in-R} in~\eqref{eq:lim-k-ball-compl} yields~\eqref{eq:suff-for-prop:mgf} and the proof is concluded.

      \item[Step 2: Proof in the case where~$b$ has compact support.]
    	In view of~\eqref{eq:disint}, it suffices to show that
    	\begin{equation}
    		g^{\alpha}(x)\int_{\cCt} \Exp{-\alpha \ss_t }\d\cP^{x}_t
    			=  \big(\Exp{t \Lambda^{\alpha} } g^{\alpha} \big)(x)
    	\end{equation}
    	for almost all~$x \in \rr^N$,
      where~$(\Exp{t \Lambda^{\alpha}})_{t\geq 0}$ is the positivity-preserving semigroup generated by~$\Lambda^{\alpha}$ on~$\mathrm{L}^{p_\alpha}(\rr^N,\d\mu_0)$.
			By definition of~$\ss_t$, this is equivalent to
			\begin{equation}
    		g^{\alpha}(x)\int_{\cCt} g^{-\alpha}(\gamma(0)) g^{\alpha}(\gamma(t))\Exp{-\alpha \epsilon^{-1} \int_0^t \braket{b(\gamma),\circ\d\gamma}}\,\cP^{x}_t(\d\gamma)
    			=  \big(\Exp{t \Lambda^{\alpha} } g^{\alpha} \big)(x).
    	\label{eq:delta-ic-mgf}
    	\end{equation}
			Note that the terms $g^\alpha(x)$ and~$g^{-\alpha}(\gamma(0))$ cancel each other out.

    	By Lebesgue monotone convergence and continuity of~$\Exp{t\Lambda^\alpha}$, it is enough to show that
    	\begin{equation}
    		\int_{\cCt} (\eta g^{\alpha})(\gamma(t))\Exp{-\alpha \epsilon^{-1} \int_0^t \braket{b(\gamma),\circ\d\gamma}}\,\cP^{x}_t(\d\gamma)
    			=  \big(\Exp{t \Lambda^{\alpha} } \eta g^{\alpha} \big)(x)
    	\end{equation}
    	for all smooth functions~$0 \leq \eta \leq 1$ with compact support. We will not keep this cutoff function~$\eta$ explicitly in the formulas, but we will use theorems that would generally apply to a continuous compactly supported function~$g$ with the understanding that we can obtain the final result by taking a sequence~$(\eta_R)_{R \in \nn}$ converging pointwise to the constant function~$1$ from below.

    	Set~$m(s,x)$ to be the left-hand side of~\eqref{eq:delta-ic-mgf} with~$t$ replaced by $s \in [0,t]$.
			Because~$g^\alpha \in \mathrm{L}^p(\rr^N,\d\mu_0^\epsilon)$
			and because the semigroup generated by~$\Lambda^\alpha$ with domain $\mathrm{W}^{2,p}(\rr^N,\d\mu_0^\epsilon)$ on $\mathrm{L}^p(\rr^N,\d\mu_0^\epsilon)$ is analytic,
			$\Exp{s\Lambda^\alpha} g^\alpha \in \mathrm{W}^{2,p}(\rr^N,\d\mu_0^\epsilon)$ and \[
				\partial_s (\Exp{s\Lambda^\alpha}g^\alpha) = \Lambda^\alpha \Exp{s\Lambda^\alpha} g^\alpha
			\]
			for all~$s > 0$; see e.g.~Proposition 1.6.ii in~\cite[Ch.\,{A-I}]{AGG+}.
			Hence,~\eqref{eq:delta-ic-mgf} becomes\ $m(t,x) = \Exp{t\Lambda^\alpha}g^\alpha$ and, by uniqueness, we need only show that~$m$ also satisfies the partial differential equation
    	 \begin{equation}\label{eq:FK-PDE}
    	 \begin{cases}
    		 \partial_s m(s,x) = (\Lambda^\alpha m(s,\cdot\,))(x), 	& x \in \rr^N, s > 0, \\
    		 m(0,x) = g^\alpha(x), 													& x \in \rr^N.
    	 \end{cases}
    	 \end{equation}
		 	 A straightforward computation shows that
     	\begin{equation}
     		\Lambda^\alpha f = \tilde \Lambda f - \tfrac{\alpha(1-\alpha)}{\epsilon} |b|^2 f + \tfrac {\alpha}{\epsilon} \braket{b,\nabla V} f - \alpha(\operatorname{div} b) f
     	\end{equation}
     	where $\tilde \Lambda$ is the infinitesimal generator associated to the deformed~\textsc{sde}
     	\begin{equation*}
     	  \d \tilde Y_t = -\nabla V(\tilde Y_t) \d t + (b(\tilde Y_t) - 2\alpha b(\tilde Y_t)) \d t + \sqrt{2\epsilon} \d W_t.
     	\end{equation*}
    	Hence, in view of the Feynman--Kac formula\,---\,see e.g.\ Lemma~3.7 in \cite[Ch.\,3]{Has} keeping in mind that~$b$ is temporarily assumed to be compactly supported\,---,~\eqref{eq:FK-PDE} will hold if
    	\begin{equation*}
        m(t,x)  = \int_{\cCt} g^\alpha (\gamma(t)) \Exp{\int_0^t - \frac{\alpha(1-\alpha)}{\epsilon} |b(\gamma(s))|^2  + \frac {\alpha}{\epsilon} \braket{b(\gamma(s)),(\nabla V)(\gamma(s))} - \alpha(\operatorname{div} b)(\gamma(s)) \d s} \tilde{\mathcal{Q}}_t^{x}(\d\gamma).
      \end{equation*}
    	But it follows from a Girsanov argument similar to that used in the proof of Proposition~\ref{lem:change-ic}\,---\,recall again that~$b$ is temporarily assumed to be compactly supported\,---\,that
    	\begin{align*}
    		&\int_{\cCt} g^\alpha (\gamma(t))  \exp\Big(\int_0^t - \tfrac{\alpha(1-\alpha)}{\epsilon} |b(\gamma(s))|^2  + \tfrac {\alpha}{\epsilon} \braket{b(\gamma(s)),(\nabla V)(\gamma(s))}
        - \alpha(\operatorname{div} b) (\gamma(s)) \d s\Big) \tilde{\mathcal{Q}}_t^{x}(\d\gamma) \\
    			&\quad=
    			\int_{\cCt} g^\alpha (\gamma(t))  \exp\Big(\int_0^t - \tfrac{\alpha(1-\alpha)}{\epsilon} |b (\gamma(s))|^2  + \tfrac {\alpha}{\epsilon} \braket{b (\gamma(s)),(\nabla V) (\gamma(s))}
					 \\ &\quad \quad \qquad
          - \alpha(\operatorname{div} b) (\gamma(s))  \d s \Big) \tilde{Z}_t(\gamma) \cP_t^{x}(\d\gamma),
    	\end{align*}
    	where
    	$
    		\tilde{Z}_s(\gamma) := \exp(\frac{\alpha}{\sqrt{2\epsilon}} \int_0^s \braket{b(\gamma),\d W(\gamma)} - \frac{\alpha^2}{4\epsilon}\int_0^s |b(\gamma(r))|^2\d r ).
    	$
    	The proof is concluded with a standard It\^o-calculus computation.	\qedhere
  \end{description}
\end{proof}

\section{Large deviations in the large-time limit}
\label{sec:large-t}

{With the validity of the formula at the heart of Section~\ref{ssec:rep-mgf} at hand under the assumptions of Section~\ref{ssec:b-term-assumptions} introduced to deal with the unboundedness of~$b$, our proof of the local large deviation principle follows closely a local analogue of the strategy outlined by J.~Lebowitz and H.~Spohn in~\cite[\S{5}]{LS99} and also carried out in~\cite[App.\,A]{BD15}.} The quantity
\[
	e^\epsilon(\alpha) := \sup \{ \Re z : z \in \sp (\Lambda^{\alpha,\epsilon},\mathrm{W}^{2,2}(\rr^N,\d\mu_0^\epsilon))\}
\]
for~$\alpha \in \cA$ will play a crucial role in this strategy for analyzing the large deviations of~$\ss_t^\epsilon$. We will interchangeably refer to this quantity as \emph{the leading eigenvalue} of~$\Lambda^{\alpha,\epsilon}$ or as $\spb(\Lambda^{\alpha,\epsilon})$.
{Before we state and prove a lemma concerning its regularity in~$\alpha$ at fixed~$\epsilon > 0$, let us briefly comment on the choice of strategy. In simple enough systems\,---\,{e.g.}\ finite-state, mixing Markov chains\,---, we are aware of two other routes. First, one can sometimes prove a higher-level \textsc{ldp} for currents/jumps and use a suitable contraction principle; see~\cite[\S{5}]{BCFG18} and~\cite[\S{2.1}]{EPbrief}. Second, one can sometimes prove the \textsc{ldp} for entropy production via the method of Ruelle--Lanford functions~\cite{CJPS19}. However,\,---\,to our knowledge\,---\,key technical ingredients that are essential to rigorously using those methods have not been adapted to stochastic integrals with respect to paths of diffusions in noncompact spaces. These difficulties are for example addressed in Remark~1 in~\cite{KKT09}, Remark~2.16 in~\cite{CJPS19}, and more precisely throughout the discussions at the end of Sections~1.6 and~2.4.3 of~\cite{EPbrief}.}

\begin{lemma}
\label{lem:reg-eigenval}
	Under Assumptions~\textnormal{(L0)}, \textnormal{(L1)} and~\textnormal{(RB)}, the function~$e^\epsilon$ is real-analytic on~$\cA$.
\end{lemma}

\begin{proof}
	Fix $\alpha_0 \in \cA$. The differential operator~$\Lambda^{\alpha_0,\epsilon}$ defined by~\eqref{eq:def-Lambda} on the domain $\mathrm{W}^{2,2}(\rr^N,\d\mu_0^\epsilon)$ is closed as an unbounded operator on~$\mathrm{L}^2(\rr^N,\d\mu_0^\epsilon)$; see Appendix~\ref{app:semig}.
	For~$\varkappa \in \cc$,
	\begin{align*}
		B^{\alpha_0,\epsilon}(\varkappa) := -2\varkappa\braket{b,\nabla} - \tfrac{\varkappa(1-\varkappa-2\alpha_0)}{\epsilon}|b|^2 + \tfrac{\varkappa}{\epsilon} \braket{b,\nabla V} - \varkappa\operatorname{div} b
	\end{align*}
	is a relatively bounded perturbation of~$\Lambda^{\alpha_0,\epsilon}$.
	The relative bound can be made arbitrarily small by taking~$|\varkappa|$ small enough.

	Hence, by Theorem~1.1 in~\cite[Ch.\,{IV}]{Kat},
	there exists a complex neighbourhood~$\Omega$ of~$\alpha_0$ such that the differential operator~$\Lambda^{\alpha,\epsilon} = \Lambda^{\alpha_0,\epsilon} + B^{\alpha_0,\epsilon}(\alpha-\alpha_0)$ on~$\mathrm{L}^2(\rr^N,\d\mu_0^\epsilon)$ with domain $\mathrm{W}^{2,2}(\rr^N,\d\mu_0^\epsilon)$ is closed for all~$\alpha \in \Omega$.
	Moreover, a straightforward estimate shows that $\varkappa \mapsto B^{\alpha,\epsilon}(\varkappa) f \in \mathrm{L}^2(\rr^N,\d\mu_0^\epsilon)$ is holomorphic whenever $f \in \mathrm{W}^{2,2}(\rr^N,\d\mu_0^\epsilon)$.
	Hence, for fixed~$\epsilon > 0$,~$\{\Lambda^{\alpha,\epsilon}\}_{\alpha \in \Omega}$ is a holomorphic family of type~(A) in the sense of~\cite[\S{VII.2.1}]{Kat}.
	By Proposition~\ref{prop:evec}, $e^\epsilon(\alpha_0)$ is a simple eigenvalue of~$\Lambda^{\alpha_0,\epsilon}$ and can be separated from the rest of~$\sp \Lambda^{\alpha_0,\epsilon}$ by a simple closed curve.
	Following~\cite[\S{VII.2.3}]{Kat}, the spectrum of~$(\Lambda^{\alpha,\epsilon},\mathrm{W}^{2,2}(\rr^N,\d\mu_0^\epsilon))$ is likewise separated into two parts for~$\alpha \in \Omega$ close enough to~$\alpha_0$, and $\alpha \mapsto e^\epsilon(\alpha)$ admits an analytic extension to a small complex neighbourhood of~$\alpha_0$.
\end{proof}

\begin{lemma}\label{lem:HF}
	Under Assumptions~\textnormal{(L0)}, \textnormal{(L1)} and~\textnormal{(RB)},
	\[
		\mathfrak{m}^\epsilon = -De^\epsilon(0).
	\]
\end{lemma}

\begin{proof}
	With the appropriate normalisation, the eigenvector corresponding to the eigenvalue~$e^\epsilon(0) = 0$ is the constant 1 and the corresponding eigenvector of the adjoint (the Fokker--Planck operator) is obtained from~$\lambda_\textnormal{inv}^\epsilon$; see the proof of Lemma~\ref{lem:pre-ub-ic-RN}. Because~$e^\epsilon$ is analytic in~$0$ and is a simple eigenvalue for all~$\alpha$ close enough to~$0$, the derivative can be computed using a formula colloquially known as the Hellmann--Feynman formula:
	\[
		De^\epsilon(0) = \int (-2\braket{b,\nabla} - \epsilon^{-1}|b|^2 + \epsilon^{-1}\braket{b,\nabla V} - \operatorname{div} b) 1 \d\lambda_\textnormal{inv}^\epsilon;
	\]
	see~(2.33) in~\cite[\S{II.2.2}]{Kat} and the argument in~\cite[\S{VII.1.3}]{Kat}.
\end{proof}

\begin{proposition}\label{prop:lim-t}
	Suppose that Assumptions~\textnormal{(L0)},~\textnormal{(L1)},~\textnormal{(RB)} and~\textnormal{(IP)} are satisfied. Then,
	\begin{equation}\label{eq:lim-t}
		\lim_{t\to\infty} \frac 1t \log \chi_t^{\epsilon}(\alpha) = e^\epsilon(\alpha)
	\end{equation}
	for all~$\alpha \in I^\epsilon$.
\end{proposition}

\begin{proof}
	Fix~$\alpha \in I^\epsilon$ and pick $p = p_\alpha^\epsilon$ as in~\textnormal{(IP)}. Let~$\psi^{\alpha,\epsilon}$ [resp.~$u^{\alpha,\epsilon}$] be a~strictly positive right [resp.~left] eigenvector of~$\Lambda^{\alpha,\epsilon}$ for the eigenvalue~$e^\epsilon(\alpha)$ with the properties of Proposition~\ref{prop:evec}.
	By Proposition~\ref{prop:mgf}, we have
	\begin{align*}
		\int_{\cCt} \Exp{-\alpha \ss^{\lambda,\epsilon}_t }\d\cP^{\xi,\epsilon}_t
			& = \int_{\rr^N} \frac{\d\lambda}{\d\mu_0^\epsilon}(x) g^{-\alpha}(x) \big(\Exp{t \Lambda^{\alpha} } g^{\alpha} \big)(x) \,\mu_0^\epsilon(\d x) \\
			& = \Exp{t e^{\epsilon}(\alpha)}
				\int_{\rr^N} \frac{\d\lambda}{\d\mu_0^\epsilon}(x) g^{-\alpha}(x)  \big(\Exp{-t e^{\epsilon(\alpha)}}\Exp{t \Lambda^{\alpha} } g^{\alpha}  - \psi^{\epsilon,\alpha} (u^{\alpha,\epsilon},g^{\alpha})_{\mu_0^\epsilon} \big)(x) \,\mu_0^\epsilon(\d x)
			 \\&\qquad\qquad{}
      + \Exp{t e^{\epsilon}(\alpha)}J^{\alpha,\epsilon},
	\end{align*}
	where~$J^{\alpha,\epsilon}$ is finite, strictly positive and independent of~$t$. Recall that our choice of~$\alpha \in I^\epsilon$ satisfying condition~\eqref{eq:I} guarantees
	\[
		g^{\alpha}(x) \in \mathrm{L}^p(\rr^N,\d\mu_0^\epsilon)
			\quad \text{and} \quad
		\tfrac{\d\lambda}{\d\mu_0^\epsilon} g^{-\alpha}(x) \in \mathrm{L}^p(\rr^N,\d\mu_0^\epsilon)^*.
	\]
	Hence, using H\"older's inequality and Proposition~\ref{prop:evec} to control the difference in the integrand,
	\[
		\lim_{t\to\infty} \frac 1t \log \int_{\cCt} \Exp{-\alpha \ss^{\lambda,\epsilon}_t }\d\cP^{\xi,\epsilon}_t = e^{\epsilon}(\alpha).
	\]
	This is exactly the property of~$\chi^\epsilon_t$ that was to be proved.
\end{proof}

\begin{remark}
	In particular, in this regime, the mean canonical entropy production and the Chernoff and Hoeffding error exponents for the hypothesis testing of the arrow of time do not depend on the specific choice of initial distribution~$\lambda$, as long as it is mutually absolutely continuous with respect to~$\mu_0^\epsilon$. Actually, if one is solely interested in this fact, one only needs the proposition for~$\alpha \in [0,1]$ and can therefore relax Assumption~\textnormal{(IP)}. We refer the reader to~\textnormal{\cite[\S{6}]{JOPS12} }and~\textnormal{\cite[\S{1.7}]{EPbrief}}.
\end{remark}

\begin{corollary}\label{cor:symm}
	Under the same assumptions, the function $e^\epsilon : \cA \to \rr$  is convex and
	\[
		e^\epsilon(1-\alpha) = e^\epsilon(\alpha)
	\]
	for all~$\alpha \in \cA$.
\end{corollary}

\begin{proof}
	Consider the particular case $g \equiv 1$ and $\lambda = \mu_0^\epsilon$ and take the appropriate limit in the second part of Proposition~\ref{prop:RN} using Proposition~\ref{prop:lim-t}.
\end{proof}

For $\varsigma \in \{ -D e^\epsilon(\alpha) : \alpha \in \cA\}$, set
\begin{equation}
	\label{eq:Legendre}
	e^\epsilon_*(\varsigma) := \sup_{\alpha \in \cA} \big(-\alpha \varsigma - e^\epsilon(\alpha)\big).
\end{equation}
It is immediate from Corollary~\ref{cor:symm}, the symmetry~$\cA = 1-\cA$ and the definition of~$e_*^\epsilon$ that
\begin{equation}
	\label{eq:CG-e}
	e^\epsilon_*(\varsigma) - e^\epsilon_*(-\varsigma) = -\varsigma
\end{equation}
for all~$\varsigma \in \{ -D e^\epsilon(\alpha) : \alpha \in \cA\}$.
Combining Lemma~\ref{lem:reg-eigenval}, Proposition~\ref{prop:lim-t} and a local version of the G\"artner--Ellis theorem (see e.g.~\cite[\S{A.2}]{JOPP}), we get the following large deviation result. The symmetry~\eqref{eq:CG-e} of the rate function~$e_*^\epsilon$ in this \textsc{ldp} is referred to as the \emph{Gallavotti--Cohen symmetry}.

\begin{proposition}\label{prop:ldp-t}
	Under assumptions~\textnormal{(L0)}, \textnormal{(L1)}, \textnormal{(RB)} and \textnormal{(IP)}, if $E$ is a Borel set with $\cl(E) \subset
	\{-De^\epsilon(\alpha) : \alpha \in I^\epsilon \}$, then
	\begin{align*}
		-\inf_{\varsigma \in \interior(E)} e^\epsilon_*(\varsigma)
			&\leq \liminf_{t \to \infty} t^{-1} \log \cP_t^\epsilon \Big\{ t^{-1} \ss^\epsilon_t \in E \Big\} \\
			&\leq \limsup_{t \to \infty} t^{-1} \log \cP_t^\epsilon \Big\{ t^{-1} \ss^\epsilon_t \in E \Big\} \leq -\inf_{\varsigma \in \cl(E)} e^\epsilon_*(\varsigma).
	\end{align*}
\end{proposition}

{By a standard argument, the above results imply that, under assumptions~\textnormal{(L0)}, \textnormal{(L1)}, \textnormal{(RB)}, \textnormal{(IP)} and $\mathfrak{m}^\epsilon > 0$, the weak law of large numbers of Proposition~\ref{prop:lln-t} occurs with exponentially fast convergence and can thus be strengthened to a strong law of large numbers.}

\section{The linear case}\label{sec:lin}

We have shown in Section~\ref{sec:large-t} that the large deviations of~$\ss_t^\epsilon$ can be understood in terms of the leading eigenvalue~$e^\epsilon(\alpha)$ of~$\Lambda^{\alpha,\epsilon}$ and its Legendre transform~\eqref{eq:Legendre}. We devote the present section to the study of these quantities in the case where we make the additional assumptions that~$V$ is quadratic and~$b$ is linear\,---\,equivalently~$V(x) = \tfrac 12 \braket{x,D^2Vx}$ and $b(x) = Db\, x$ up to a shift in the space coordinates. Note that~\textnormal{(ND)} is then a consequence of~\textnormal{(L0)}, which becomes
\begin{equation}\label{eq:Cun-lin-case}
	D^2 V > 0.
\end{equation}
Assumption~\textnormal{(RB)} becomes
\begin{equation}\label{eq:Ctrois-lin-case}
	\braket{Db\, x,  D^2V x} \leq k_b|D^2V x|^2
\end{equation}
for all~$x\in\rr^N$.

The linear case is particularly important for several reasons. First and foremost, we will see in Sections~\ref{sec:small-e} and~\ref{app:semiclass} that the general case can be reduced to this one in the limit~$\epsilon \to 0$.
Second, linearity makes computations more tractable and allows to give a characterisation of the vanishing of the mean entropy production per unit time~$\mathfrak{m}^\epsilon$.

Note that the operator~$\Lambda^{\alpha,\epsilon}$ introduced in~\eqref{eq:def-Lambda} is in this case isospectral to the $\epsilon$-independent operator
\begin{equation}\label{def:quad-op}
	Q^{\alpha} = \upDelta + \braket{\ell_{B^{(\alpha)}} , \nabla} - q_{K^{(\alpha)}} + \tfrac 12 \tr D^2V - \alpha \tr Db
\end{equation}
where $\ell_{B^{(\alpha)}}$ is the auxiliary linear vector field~$x \mapsto B^{(\alpha)}x$ and~$q_{K^{(\alpha)}}$ is the auxiliary quadratic potential~$x \mapsto \braket{x, K^{(\alpha)}x}$, with
\[
	B^{(\alpha)} := (1-2\alpha) Db
\]
and
\[
	K^{(\alpha)} := \tfrac 14 (D^2V)^2 - \tfrac 14 (Db^\mathsf{T}D^2V + D^2VDb) + \alpha(1-\alpha) Db^\mathsf{T}Db.
\]
To see this, conjugate with the Gaussian weight $\Exp{-(2\epsilon)^{-1}V}$ and its inverse and then make a change of variable~$x \mapsto \epsilon^{1/2}x$.

Such elliptic operators with quadratic symbols have been fairly well understood since the seminal work of~\cite{Sj74}. Here, inspired by~\cite{FS97,JPS17}, we emphasise a slightly different point of view, which relies on the study of the corresponding \emph{algebraic Riccati equation} (\textsc{are})
\begin{equation}
\label{eq:ARE-abs-first}
	X^2 - \tfrac 12 (B^{(\alpha)})^\mathsf{T} X - \tfrac 12 X B^{(\alpha)} - K^{(\alpha)} = 0
\end{equation}
for a symmetric matrix~$X$. The general theory of such equations is discussed in~\cite{LaRo}. See~\cite{BCX20} for yet another approach in a special case.

\begin{proposition}\label{prop:eval-ARE}
	For all~$\alpha \in \cA$, the \textsc{are}~\eqref{eq:ARE-abs-first} admits a maximal solution~$X^{(\alpha)}$ and
	\[
		\spb Q^{\alpha} = -\tr X^{(\alpha)} + \tfrac 12 \tr D^2V - \alpha \tr Db.
	\]
	Moreover, $\alpha \mapsto \tr X(\alpha)$ defines a real-analytic function on~$\cA$ and we have the identity
	\begin{equation}
		\tr X^{(\alpha)}  = -\frac 12 \bigg( \tr B^{(\alpha)} -  \sum_{\lambda^{(\alpha)} \in \sp \mathcal{K}_\textnormal{Ham}^{(\alpha)}} |\Re \lambda^{(\alpha)}| \bigg)
	\end{equation}
	where
	\begin{equation}\label{eq:Ham-matrix}
		\mathcal{K}_\textnormal{Ham}^{(\alpha)}
		:=
		\left[
		\begin{matrix}
			-\tfrac 12 B^{(\alpha)}		& \one \\
			K^{(\alpha)}								& \tfrac 12 (B^{(\alpha)})^\mathsf{T}
		\end{matrix}
		\right].
	\end{equation}
\end{proposition}

\begin{proof}
  Consider~$\phi_X(x) := \exp (-\tfrac 12 \braket{x,Xx})$ for some positive-definite matrix~$X$ and compute
	\begin{align*}
		(Q^\alpha \phi_X)(x) &=  - \tr X \ \phi_X(x) + \braket{Xx,Xx} \phi_X(x) - \braket{B^{(\alpha)}x,Xx}\phi_X(x)
		\\ & \qquad {}
		- \braket{x,K^{(\alpha)}x}\phi_X(x) + (\tfrac 12 \tr D^2V - \alpha \tr Db)\phi_X(x).
	\end{align*}
	Note that~$\phi_X$ is an eigenvector with eigenvalue $-\tr X + \tfrac 12 \tr D^2V - \alpha \tr Db$ if
	\begin{equation}
	\label{eq:ARE-abs}
		R(\alpha,X) := X^2 - \tfrac 12 (B^{(\alpha)})^{\mathsf{T}} X - \tfrac 12 X B^{(\alpha)} - K^{(\alpha)} = 0.
	\end{equation}
	Because~$K^{(\alpha)}$ is positive definite for all~$\alpha \in \cA$, $R(\alpha, 0) < 0$. Therefore, there exists a maximal positive-definite matrix~$X^{(\alpha)}$ such that~$R(\alpha,X^{(\alpha)})=0$, and~$-X^{(\alpha)} + \tfrac 12 B^{(\alpha)}$ is stable~\cite[\S{9.1}]{LaRo}.
	This argument is valid for all~$\alpha \in \cA$ and~$X^{(\alpha)}$ is a real-analytic function of~$\alpha \in \cA$~\cite[\S{11.3}]{LaRo}.

	In $\alpha = \tfrac 12 $, we have $R(\tfrac 12,X) = X^2 - K^{(1/2)}$ and the square root of~$K^{(1/2)}$
	clearly is the maximal solution to the \textsc{are}~$R(\tfrac 12,X) = 0$. But the trace of this maximal solution coincides with the smallest eigenvalue of the quantum harmonic oscillator~$-\upDelta + q_{K^{(1/2)}}$. Thus, first part of the lemma follows by simplicity and continuity of~$\spb Q^\alpha$.
	Relations between the eigenvalues of~$-X^{(\alpha)} + \tfrac 12 B^{(\alpha)}$ and those of the matrix~\eqref{eq:Ham-matrix} are discussed in~\cite[\S{8.3}]{LaRo}.
\end{proof}

\begin{remark}\label{rem:formula-is-ok-for-sp}
	Note that once a Gaussian weight is introduced to define~$Q^\alpha$, the method for obtaining the formula for its leading eigenvalue does not appeal to the fact~$D^2V > 0$, but only to the fact that~$(D^2 V)^2 > 0$.
\end{remark}

\begin{proposition}\label{prop:locally-a-grad}
	Under the assumptions of Proposition~\ref{prop:lim-t} and the additional assumption that $V$ is quadratic and~$b$ is linear,
	\begin{equation}\label{eq:lim-t-lin}
		\lim_{t \to \infty} \frac 1t \log \chi_t^\epsilon(\alpha) = -\tr X^{(\alpha)} + \tfrac 12 \tr D^2V - \alpha \tr Db,
	\end{equation}
	for all~$\alpha \in \cA$. Moreover,
	\begin{enumerate}
		\item[i.] if the matrix~$Db$ is not symmetric, then the mean entropy production per unit time~$\mathfrak{m}^\epsilon$ is strictly positive and independent of~$\epsilon$ and the rate function~$e^\epsilon_*$ in Proposition~\ref{prop:ldp-t} is strictly convex and independent of~$\epsilon$;
		\item[ii.] if the matrix~$Db$ is symmetric, then~$\mathfrak{m}^\epsilon = 0$.
	\end{enumerate}
\end{proposition}

\begin{proof}
	Combining Proposition~\ref{prop:lim-t} and Proposition~\ref{prop:eval-ARE} with the fact that
	\[
		e^\epsilon(\alpha) = \spb Q^\alpha
	\]
	immediately gives~\eqref{eq:lim-t-lin}. It follows from Corollary~\ref{cor:symm} that~$e^\epsilon$ is convex on~$\cA$ and that~$e^\epsilon(0) = e^\epsilon(1) =0$.
	Hence, by analyticity, it will fail to be strictly convex if and only if it vanishes identically, which is in turn equivalent to~$e^\epsilon(\tfrac 12) = 0$.
	This last condition takes the explicit form
	\[
		\tr\sqrt{(D^2V - Db)^\mathsf{T}(D^2V-Db)} = \tr(D^2V-Db).
	\]
	Let $A := D^2V-Db$ and $|A| := \sqrt{A^\mathsf{T}A}$. We can find orthonormal bases~$\{v_i\}_{i=1}^N$ and~$\{w_i\}_{i=1}^N$ of~$\cc^N$ such that $A = \sum_{i=1}^N \mu_i v_i \braket{w_i,\cdot\,}$ and $|A|  = \sum_{i=1}^N \mu_i w_i \braket{w_i,\cdot\,}$,
	where $\{\mu_i\}_{i=1}^N$ are the singular values of~$A$ listed with multiplicity; see e.g.\ \cite[\S{3.5}]{Sim4}.
	Computing traces in the basis~$\{w_i\}_{i=1}^N$ and using $\mu_i \geq 0$, we find that $\tr A = \tr |A|$ implies
	$\braket{v_i,w_i} = 1$ for each~$i$ such that $\mu_i \neq 0$. Because~$|w_i| = |v_i| = 1$, $\braket{v_i,w_i} = 1$ implies $w_i = v_i$ and we conclude that $A = |A|$. Of course, $A = |A|$ implies $\tr A = \tr |A|$.

	Since $D^2V$ is already symmetric, $A = |A|$ if and only if $Db^\mathsf{T} = Db$ and all the eigenvalues of $D^2V - Db$ are nonnegative.
	For the second condition only to fail, we would need a nonzero vector~$u$ and a strictly positive number~$\lambda$ such that $(D^2V-Db)u = -\lambda u$. Taking an inner product with $D^2V u$ in this eigenvalue equation gives
	\begin{align*}
		\braket{Db \, u, D^2V u} &= \big|D^2Vu\big|^2 + \lambda \braket{u, D^2V u},
	\end{align*}
	would then contradict~\eqref{eq:Cun-lin-case}--\eqref{eq:Ctrois-lin-case}.
\end{proof}

Note that Case~i in Proposition~\ref{prop:locally-a-grad} occurs if and only if the linear vector field~$b$ is nonconservative; Case~ii, if~$b$ is conservative. To see this, recall that the Hessian of a sufficiently regular function is always symmetric and that the gradient of a function of the form~$x \mapsto \tfrac 12 \braket{x,Bx}$ is the linear vector field $x \mapsto \tfrac 12 (B + B^{\mathsf{T}})x$. In view of this, we will say that a nonlinear vector field~$b$ ``behaves like a gradient'' near a point~$x$ if~$Db|_x$ is is symmetric.

\section{The rate function in the vanishing-noise limit}
\label{sec:small-e}

We consider the limit~$\epsilon \to 0$. The main result of this section is the local \textsc{ldp} of Theorem~\ref{prop:ldp-t-eps}, but we also discuss the behaviour of the mean entropy production per unit time. It is reasonable to allow the initial condition~$\lambda$ and the function~$g$ to change with~$\epsilon$\,---\,it is in fact necessary if one wants to study the steady-state canonical entropy production. We require Assumption~\textnormal{(IP)} to hold with a certain uniformity in~$\epsilon$.
\begin{description}
	\item[Assumption (IPu).] There exists an open interval~$I^0$ containing~$0$ and~$1$, and whose closure is contained in $\liminf_{\epsilon\to 0} I^\epsilon$, where~$I^\epsilon$
	is as in Assumption~\textnormal{(IP)} with $g$ replaced with $g^\epsilon$ and $\lambda$ replaced with~$\lambda^\epsilon$.
\end{description}
Before we proceed to the general statements and proofs, let us illustrate the main points with an example.

\begin{example}\label{ex:orth-0-div}
	Let $V$ be a potential satisfying our general assumptions and suppose that its global minimum is achieved in a single point~$x_{j^\star}$.
	Suppose that $b$ satisfies our general assumptions as well as $\operatorname{div} b \equiv 0$ and $\braket{b,\nabla V} \equiv 0$, and consider the steady-state functional with~$g \equiv 1$. This is a situation in which one can easily show that~$\lambda = \lambda_\textnormal{inv}^\epsilon = \mu_0^\epsilon$.

	At the level of the mean entropy production per unit time, one can
	show the convergence $\mathfrak{m}^\epsilon \to \mathfrak{m}_{j^\star}$,  where~$\mathfrak{m}_{j^\star}$ is as in Section~\ref{sec:lin} for the linear problem near~$x_{j^\star}$. In particular, we have strict positivity of the limit if and only if~$b$ does not behave like a gradient near~$x_{j^\star}$. This strict positivity is a key signature of nonequilibrium.

	At the level of the fluctuations, the situation is the following. If~$|\alpha|$ is small enough, $e^\epsilon(\alpha) \to \max_j e_j(\alpha)$, where the maximum is taken over indices~$j$ corresponding to \emph{all local minima} of~$V$ and~$e_{j}$ is as in Section~\ref{sec:lin} for the linear problem near~$x_{j}$.
	Therefore, with~$e_*$ the Legendre transform of~$\alpha \mapsto \max_j e_j(\alpha)$, the rate functions~$e^\epsilon_*(\varsigma)$ converge to~$e_*(\varsigma)$ for all~$\varsigma$ in an interval~$\Sigma$.
	In cases where there is at least one index~$j'$ corresponding to local minimum such that $D e_{j'}(0) \neq 0$, the interval~$\Sigma$ has nonempty interior.
	Hence, as far as the rate of exponential suppression of fluctuations is concerned, there is no discrimination between the global and local minima of~$V$.

	In cases where there are indices~$j'$ and~$j''$ corresponding to local minima such that $D e_{j'}(0) \neq D e_{j''}(0)$, then~$e_{j'}$ and~$e_{j''}$ cross in $\alpha = 0$. Such a crossing necessarily yields a nondegenerate closed interval strictly contained in~$\Sigma$ on which the rate function~$e_*$ vanishes. Hence, by tuning the behaviour of~$b$ near the critical points of a potential~$V$ with a single global minimum and other local minima, one can construct examples where~$\lim_\epsilon \mathfrak{m}^\epsilon$ lies at either end of this vanishing piece as well as examples where it lies in the interior.
\end{example}

Back to the general case,
recall that we have successfully reduced the study of the rate function to that of the leading eigenvalue~$e^\epsilon(\alpha)$ of the deformed generator~$\Lambda^{\epsilon,\alpha}$
and its Legendre transform in the variable~$\alpha$.
Because
\begin{multline}\label{eq:conj-def-gen}
	\Exp{-(2\epsilon)^{-1}V}\Lambda^{\epsilon,\alpha}(\Exp{(2\epsilon)^{-1}V}f) \\
	= \epsilon \upDelta f +  \braket{(1 - 2\alpha)b,\nabla f} - \tfrac 1{4\epsilon}|\nabla V|^2 f
	+ \tfrac 1{2\epsilon} \braket{b,\nabla V}f  {}  - \tfrac{\alpha(1-\alpha)}{\epsilon} |b|^2f + \tfrac 12 f \upDelta V - \alpha  f \operatorname{div} b
\end{multline}
for sufficiently regular~$f$, the semiclassical folklore suggests that the quadratic approximations near the zeroes of~$\tfrac 1{4}|\nabla V|^2  - \tfrac 1{2} \braket{b,\nabla V} + \alpha(1-\alpha) |b|^2$\,---\,which coincide with the critical points of~$V$ for~$\alpha \in \cA$\,---\,should
play an important role as $\epsilon \to 0$.
While it is possible that Proposition~\ref{prop:lim-t-eps} below is known to workers in the field of semiclassical analysis, we were not able to track a convenient reference and hence provide a complete proof in Section~\ref{app:semiclass}.

Such a quadratic approximation of the deformed conjugated generator near a critical point~$x_j$ is of the form treated in Section~\ref{sec:lin}.
In view of this analysis, we define
\begin{equation}\label{eq:ej-as-trace}
  e_j(\alpha) := -\tr X_j^{(\alpha)} + \tr \tfrac 12 D^2 V|_{x_j} -\alpha \tr Db|_{x_j},
\end{equation}
for~$\alpha \in \cA$, where~$X_j^{(\alpha)}$ is the maximal solution to the~\textsc{are}
\begin{equation}\label{eq:ARE-main-matt}
	\big(X_j^{(\alpha)}\big)^2 - \tfrac 12 \big(B_j^{(\alpha)}\big)^\mathsf{T} X_j^{(\alpha)} - \tfrac 12 X_j^{(\alpha)} B_j^{(\alpha)} - K_j^{(\alpha)} = 0
\end{equation}
with $B_j^{(\alpha)} := (1-2\alpha) Db|_{x_j}$
and
\[
  K_j^{(\alpha)} := \tfrac 14 D^2 V|_{x_j}^2 - \tfrac 14 (Db|_{x_j}^\mathsf{T} D^2V|_{x_j} + D^2V|_{x_j}^{\mathsf{T}}Db|_{x_j}) + \alpha(1-\alpha) Db|_{x_j}^{\mathsf{T}}Db|_{x_j}.
\]
We give an example in Figure~\ref{fig:cgf-sketch}.

\begin{figure}
	\begin{center}
  		\includegraphics{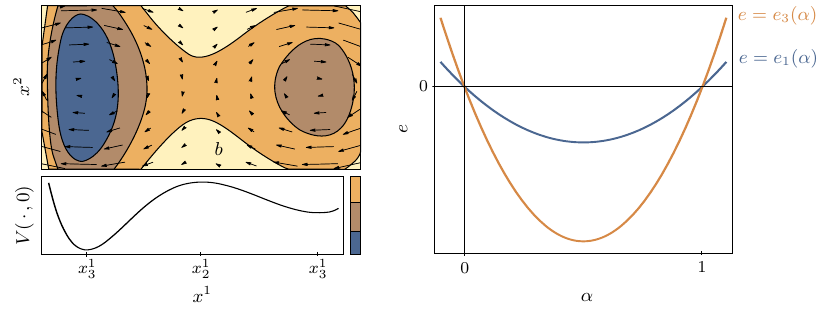}
	\end{center}
  \caption{We consider a polynomial potential~$V : \rr^2 \to \rr$ with a global maximum in~$x_1 = (x_1^1,0)$, a saddle point in~$x_2 = (x_2^1,0)$ and a local minimum in~$x_3 = (x_3^1,0)$. On the left: the profile of~$V$ for~$x^2 \equiv 0$ as well as a nonconservative vector field~$b$ which is stationary in all those critical points superimposed on a contour plot of~$V$.
	On the right:~$e_1$ and~$e_3$ from~\eqref{eq:ej-as-trace} are plotted as functions of~$\alpha$;~$e_2$ lies below the visible region.}
  \label{fig:cgf-sketch}
\end{figure}

\begin{lemma}\label{lem:single-out-stable}
  Suppose that Assumptions~\textnormal{(RB)} and~\textnormal{(ND)} are satisfied. Then, $e_j(0) \leq 0$ with equality if and only if~$x_j$ is a local minimum of~$V$.
\end{lemma}

\begin{proof}
  One can check directly that $\tfrac 12 D^2V|_{x_j}$ is~\emph{a} symmetric solution to~\eqref{eq:ARE-main-matt} with~$\alpha = 0$, so that $X_j^{(0)} \geq \tfrac 12 D^2V|_{x_j}$ and
  \begin{equation}\label{eq:ARE-or}
    e_j(0) = -\tr X_j^{(0)} + \tfrac 12 \tr D^2V|_{x_j} \leq 0.
  \end{equation}
  On the other hand, Assumption~\textnormal{(RB)} yields that the matrix~$0$ is a subsolution to~\eqref{eq:ARE-main-matt} with~$\alpha = 0$, which implies that
	\[
		X_j^{(0)} \geq 0.
	\]
	If~$x_j$ is not local minimum, then~$D^2V|_{x_j}$ is not positive semidefinite by~\textnormal{(ND)} and the inequality~\eqref{eq:ARE-or} must be strict.
\end{proof}

\begin{proposition}\label{prop:lim-t-eps}
	Suppose that Assumptions~\textnormal{(L0)},~\textnormal{(L1)},~\textnormal{(RB)} and~\textnormal{(ND)} are satisfied. Then, for all~$\alpha \in \cA$,
	\begin{equation}\label{eq:lim-eps-as-max}
		\lim_{\epsilon\to 0} e^\epsilon(\alpha) = \max_{j=1,\dotsc,m} e_j(\alpha).
	\end{equation}
	 The convergence is uniform on compact subsets of~$\cA$. The limit defines a convex and piecewise real-analytic function~$e : \cA \to \rr$ satisfying the symmetry $e(1-\alpha) = e(\alpha)$, and~$D e^\epsilon(\alpha)$ converges to~$D e(\alpha)$ for all~$\alpha$ in a dense subset of~$\cA$
\end{proposition}

\begin{proof}
	By~\eqref{eq:conj-def-gen}, Proposition~\ref{prop:eval-ARE} and Remark~\ref{rem:formula-is-ok-for-sp},
	\begin{equation}
		\label{eq:max-e-Q}
		\max_{j=1,\dotsc,m} e_j(\alpha) = \max_{j=1,\dotsc,m} \spb Q_j^\alpha
	\end{equation}
	for all~$\alpha\in\cA$, where $Q_j^\alpha$ has the form
	\begin{equation}\label{eq:def-Qj}
		Q_j^\alpha := \upDelta + \braket{\ell_{B_j^{(\alpha)}}, \nabla} - q_{K_j^{(\alpha)}} + \tfrac 12 \tr D^2V|_{x_j} - \alpha \tr Db|_{x_j}.
	\end{equation}
	We postpone the proof of the fact that
	\begin{equation}\label{eq:semiclass-lim-to-prove}
		\lim_{\epsilon\to 0} e^{\epsilon}(\alpha) = \max_{j=1,\dotsc,m} \spb Q_j^\alpha
	\end{equation}
	to Section~\ref{app:semiclass}.

	Let~$\bar{J}$ be a compact subset of~$\cA$. The fact that the convergence is uniform on~$\bar{J}$ and that the derivatives converge on a dense subset are well-known consequences of convexity. Each~$e_j$ is real analytic on~$\cA$ by Proposition~\ref{prop:eval-ARE}. Hence, the difference between any two~$e_j$ and~$e_{j'}$ is real analytic and therefore has finitely many zeroes on~$\bar{J}$, or~$e_j \equiv e_{j'}$ on~$\cA$. It is no loss of generality to exclude the second case. There must be at most finitely many points in~$\bar{J}$ where the maximum in~\eqref{eq:max-e-Q} changes index. We conclude that~$e$ is piecewise real analytic.
\end{proof}

Proposition~\ref{prop:lim-t-eps} has the following important consequences.
Note that Lemma~\ref{lem:single-out-stable} implies that the maximum in Proposition~\ref{prop:lim-t-eps} must be achieved for an index~$j$ corresponding to a local minimum if~$\alpha$ is close enough to~$0$.
Thus, using Lemma~\ref{lem:HF},
\begin{equation}\label{eq:rough-bound-m}
	\min_{j\textnormal{ loc.\,min.}} \mathfrak{m}_j
		\leq \liminf_{\epsilon \to 0} \mathfrak{m}^\epsilon
		\leq \limsup_{\epsilon \to 0} \mathfrak{m}^\epsilon
		\leq \max_{j\textnormal{ loc.\,min.}} \mathfrak{m}_j,
\end{equation}
where
\begin{equation}\label{eq:def-moy-j}
	\mathfrak{m}_j := - De_j(\alpha)\big|_{\alpha=0}
\end{equation}
for indices~$j$ that correspond to local minima of~$V$. The fact that we are not able to generally strengthen~\eqref{eq:rough-bound-m} by taking the minimum and maximum only over indices corresponding to~\emph{global} minimisation of~$V$ as in Example~\ref{ex:orth-0-div}
is a drawback of the freedom of the decomposition mentioned in Remark~\ref{rem:orth-decomp}.
		To see this, consider a potential~$V$ with its global minimum achieved in two points~$x_{j^\star}$ and~$x_{j^{\star\star}}$.
		The changes $V \mapsto V + \delta \, \eta_{\star}$ and $b \mapsto b + \delta \nabla \eta_{\star} $ for a small positive number~$\delta$ and a suitable bump function~$\eta_\star$ centered at~$x_{j^\star}$ do not change the dynamics nor the validity of the assumptions, but the new potential does not achieve its global minimum in~$x_{j^\star}$.
		Such a freedom is gone if we restrict are attention to decompositions where the Freidlin--Wentzell quasipotential~\cite{FW70} is proportional to~$V$\,---\,as is the case in Example~\ref{ex:orth-0-div}.

Recall that Proposition~\ref{prop:locally-a-grad} gives that~$\mathfrak{m}_j$ in~\eqref{eq:def-moy-j} is nonnegative and equals zero if and only if~$Db|_{x_j}$ is symmetric. Therefore, the mean entropy production per unit time~$\mathfrak{m}^\epsilon$ vanishes as~$\epsilon \to 0$ if~$b$ behaves like a gradient near each local minimum of~$V$.
On the other hand,~$\mathfrak{m}^\epsilon$ is bounded away from 0 as $\epsilon \to 0$ if there is no local minimum of~$V$ near which~$b$ behaves like a gradient.
From a thermodynamical point of view, strict positivity of the mean entropy production per unit time~$\mathfrak{m}^\epsilon$ is a key signature of nonequilibrium.

The nonvanishing of~$\mathfrak{m}^\epsilon$ also ensures that the content of our \textsc{ldp} is nontrivial. Indeed, the intervals
\[
	\Sigma := \liminf_{\epsilon \to 0}
	\{-De^\epsilon(\alpha) : \alpha \in \cA\}
\]
and
\[
	\Sigma^0 := \liminf_{\epsilon \to 0}\{-De^\epsilon(\alpha) : \alpha \in I^0\},
\]
are always nonempty, but could {a priori} be singletons; strict positivity of~$\mathfrak{m}^\epsilon$ in the limit $\epsilon \to 0$ rules out this possibility. More generally, degeneracy of these intervals is ruled out whenever there exist a local minimum of~$V$ near which~$b$ does not behvae like a gradient.

\begin{proposition}\label{prop:lim-rate-e}
	Suppose that Assumptions~\textnormal{(L0)},~\textnormal{(L1)},~\textnormal{(RB)},~\textnormal{(ND)} and~\textnormal{(IPu)} are satisfied. If~$E$ is a Borel set with $\cl(E) \subset \interior(\Sigma)$, then
	\begin{equation}\label{eq:lim-e-Leg-with-inf}
		\lim_{\epsilon \to 0} \inf_{s \in E} e^\epsilon_*(\varsigma) = \inf_{s \in E} e_*(\varsigma)
	\end{equation}
	where
	\[
	 e_*(\varsigma) := \sup_{\alpha \in \cA} \big(-\alpha \varsigma - e(\alpha) \big)
	\]
	defines a convex and nonnegative function of~$\varsigma \in \Sigma$.
\end{proposition}

\begin{proof}
	The proposition is vacuously true if $\Sigma$ has empty interior. Let us now consider that $\interior(\Sigma)$ is nonempty. Convexity of~$e_*$ follows from that of~$e$.
	Since $\cl(E) \subset \interior(\Sigma)$, Proposition~\ref{prop:lim-t-eps} ensures that we may pick~$\alpha_1$ and $\alpha_2$ in~$\cA$ such that
	\[
		\inf \Sigma < -De(\alpha_1) < \inf E \leq  \sup E < -De(\alpha_2) < \sup \Sigma
	\]
	while
	\[
		\lim_{\epsilon \to 0} -D e^\epsilon(\alpha_1) = -De(\alpha_1)
		\qquad \textnormal{ and } \qquad
		\lim_{\epsilon \to 0} -D e^\epsilon(\alpha_2) =  -De(\alpha_2).
	\]
	Then, for any $\varsigma \in E$ and $\epsilon > 0$ small enough, we have
	\[
		 -De^\epsilon(\alpha_1) < \varsigma < -De^\epsilon(\alpha_2).
	\]
	Therefore,
	\[
		 e_*(\varsigma) = \sup_{\alpha \in \cA} \big(-\alpha \varsigma - e(\alpha) \big) = \sup_{\alpha \in [\alpha_2,\alpha_1]} \big(-\alpha \varsigma - e(\alpha) \big)
	\]
	and
	\[
		 e^\epsilon_*(\varsigma) = \sup_{\alpha \in \cA} \big(-\alpha \varsigma - e^\epsilon(\alpha) \big) = \sup_{\alpha \in [\alpha_2,\alpha_1]} \big(-\alpha \varsigma - e^\epsilon(\alpha) \big)
	\]
	for $\epsilon > 0$ sufficiently small. The result thus follows from the uniform convergence of~$e^\epsilon$ to~$e$ on the compact interval~$[\alpha_2,\alpha_1]$ in Proposition~\ref{prop:lim-t-eps}.
\end{proof}

The interest of Proposition~\ref{prop:lim-rate-e} of course is that it can be used in conjunction with the local \textsc{ldp} of Proposition~\ref{prop:ldp-t} for fixed~$\epsilon > 0$. The last part of the following theorem is illustrated by an example sketched in Figure~\ref{fig:rate-sketch}.

\begin{figure}
	\begin{center}
		\includegraphics{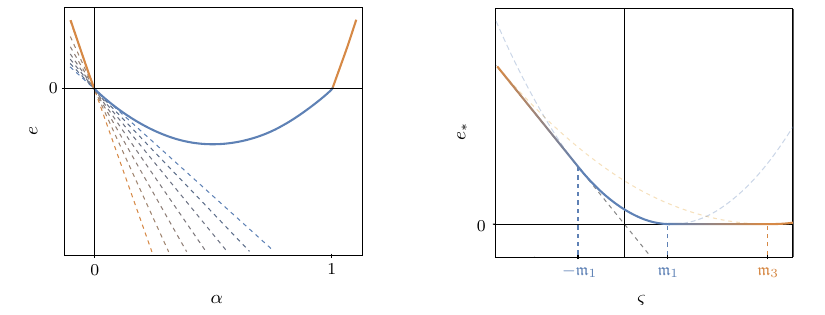}
	\end{center}
	\caption{In the case of the function~$e$ on the left, obtained by taking the maximum of~$e_1$,~$e_2$ and~$e_3$ in Figure~\ref{fig:cgf-sketch}, the jump in the derivative from $-\mathfrak{m}_3$ to~$-\mathfrak{m}_1$ at the origin causes~$e_*$ to vanish on the interval~$[\mathfrak{m}_1,\mathfrak{m}_3]$. The Legendre transform~$e_*$ is sketched on the right.}
  \label{fig:rate-sketch}
\end{figure}

\begin{theorem}\label{prop:ldp-t-eps}
	If Assumptions~\textnormal{(L0)},~\textnormal{(L1)},~\textnormal{(RB)},~\textnormal{(ND)} and~\textnormal{(IPu)} are satisfied and~$E$ is a Borel set with $\cl(E) \subset  \interior (\Sigma^0)$, then
	\begin{align*}
		-\inf_{\varsigma \in \interior(E)} e_*(\varsigma)
			&\leq \lim_{\epsilon\to 0}\liminf_{t \to \infty} t^{-1} \log \cP_t^\epsilon \Big\{ t^{-1} \ss^\epsilon_t \in E \Big\} \\
			&\leq \lim_{\epsilon\to 0}\limsup_{t \to \infty} t^{-1} \log \cP_t^\epsilon \Big\{ t^{-1} \ss^\epsilon_t \in E \Big\} \leq -\inf_{\varsigma \in \cl(E)} e_*(\varsigma).
	\end{align*}
	and the function~$e_* : \Sigma^0 \to [0,\infty)$ is continuous and satisfies the Gallavotti--Cohen symmetry
	\[
		e_*(\varsigma) - e_*(-\varsigma) = -\varsigma.
	\]
	If $\min_{j\textnormal{ loc.\,min.}}\mathfrak{m}_j \neq \max_{j\textnormal{ loc.\,min.}}\mathfrak{m}_j$, then these two values define a nondegenerate interval in~$\Sigma^0$ on which~$e_*$ vanishes.
\end{theorem}

\begin{remark}
	Recall that the rate function~$e_*$ is the Legendre transform of~$e$, which is in turn the pointwise maximum among the family~$\{e_j\}_{j=1}^m$. Therefore,~$e_*$ can be computed as the convex hull of the family~$\{(e_j)_*\}_{j=1}^m$ of Legendre transforms coming from the linearised problems near the critical points of~$V$; see Theorem~16.5 in~\textnormal{\cite[Pt.\,III]{Roc}}.
\end{remark}

\section{Convergence in the proof of Proposition~\ref{prop:lim-t-eps}}
\label{app:semiclass}

We devote this section to proving the semicalssical result at the core of Proposition~\ref{prop:lim-t-eps}, that is the convergence expressed in~\eqref{eq:semiclass-lim-to-prove} for $e^\epsilon(\alpha) := \spb \Lambda^{\epsilon,\alpha}$. Our proof of the lower bound
\[
	\liminf_{\epsilon \to 0} \spb(\Lambda^{\epsilon,\alpha}) \geq \max_{j=1,\dotsc,m} \spb(Q_j^\alpha)
\]
uses the Protter--Weinberger characterisation of the spectral bound and follows some ideas of \cite[\S{5}]{BNS94}. The Protter--Weinberger characterisation is a variational principle which states that
\[
	\spb Q_j^\alpha = \inf_{u \gg 0} \sup_{x} \frac{(Q_j^\alpha u)(x)}{u(x)},
\]
where the infimum is taken over all strictly positive function~$u$ of class~$C^2$, and similarly for other uniformly elliptic operators; see~\cite{PW66,DV75a,NP92}. Our proof of the upper bound
\[
	\limsup_{\epsilon \to 0} \spb(\Lambda^{\epsilon,\alpha}) \leq \max_{j=1,\dotsc,m} \spb(Q_j^\alpha)
\]
is inspired by B.~Simon's localisation argument in the self-adjoint case~\cite[\S{2--3}]{Si83}, with the Rayleigh--Ritz principle replaced by the Protter--Weinberger principle.

Let us mention that the selfadjoint case was also covered by B.~Helffer and J.~Sj\"ostrand in a series of papers starting with~\cite{HS84} using different methods. In the non-selfadjoint case, a collection of similar results are available, even beyond the elliptic case, but under some extra smoothness and growth conditions; see e.g.~\cite{HSS05,HPS13}. Under minimal regularity assumptions for the quadratic expansion to make sense, W.~H.~Flemming and Sh.-J.~Sheu proved a similar result in the case of a single minimum; see~\cite[\S{4}]{FS97}.

We fix~$\alpha\in\cA$ for the rest of the section and omit the corresponding superscript from the notation. We show in Appendix~\ref{app:semig} (take $p = 2$ there) that the spectral properties of~$\Lambda^{\epsilon}$ can be deduced from those of the operator
\[
	A^\epsilon := \epsilon\upDelta + \braket{F,\nabla} - \epsilon^{-1} W_0  - W_1
\]
on the space~$\mathrm{L}^2(\rr^N,\d\vol)$, with domain~
\[
	\mathsf{D}^2 := \{f \in \mathrm{W}^{2,2}(\rr^N,\d\vol) : {|\nabla V|^2} f \in \mathrm{L}^2(\rr^N,\d\vol)\},
\]
with the auxiliary vector field
\[
	F := (1-2\alpha) b
\]
and the auxiliary potentials
\begin{equation}
\label{eq:aux-pot}
	W_0 :=  \tfrac 14 |\nabla V|^2  - \tfrac 12 \braket{b,\nabla V} + \alpha(1-\alpha) |b|^2
	\qquad \text{and} \qquad
	W_1 := - \tfrac 12 \upDelta V + \alpha\operatorname{div} b.
\end{equation}
We will use the fact that,~$F$, $W_0$ and~$W_1$ are of class~$C^2$,~$C^3$ and~$C^1$ respectively, but these assumptions can be slightly relaxed if necessary.
With
$
	\ell_j(x) := DF|_{x_j}(x-x_j),
$
$
	 q_{j}(x) := \tfrac 12 \braket{x-x_j, D^2W_0|_{x_j} (x-x_j)}
$
and
$
	w_{j} := -\tfrac 12 \tr D^2V|_{x_j} + \alpha \tr Db|_{x_j}
$
for each index $j = 1,\dotsc, m$, we set
\[
	Q_j^\epsilon := \epsilon\upDelta + \braket{\ell_j,\nabla} - \epsilon^{-1} q_j - w_{j}.
\]
This is the best approximation of~$A^\epsilon$ near~$x_j$ which is of the form considered in Section~\ref{sec:lin}.
Its leading eigenvalue admits $\phi_j^\epsilon := \exp (-(2\epsilon)^{-1} \braket{x-x_j, X_j(x-x_j)})$ as an eigenvector, where $X_j$ is positive definite and satisfies the~\textsc{are}
\begin{equation}\label{eq:ARE-for-semicl}
	X_j^2 - \tfrac 12 DF|_{x_j}^\mathsf{T} X_j - \tfrac 12 X_j DF|_{x_j} = \tfrac 12 D^2W_0|_{x_j}.
\end{equation}
Note that $Q_j$ defined in~\eqref{eq:def-Qj} coincides with $Q_j^1$ and that the leading eigenvalue $\spb Q_j^\epsilon$ is independent of~$\epsilon$.

\begin{description}
	\item[Lower bound.] If Assumptions~\textnormal{(L0)},~\textnormal{(L1)},~\textnormal{(RB)} and~\textnormal{(ND)} are satisfied, then
	\[
		\liminf_{\epsilon \to 0} \spb(A^{\epsilon}) \geq \max_{j=1,\dotsc,m} \spb(Q_j).
	\]
\end{description}

	Let~$j \in \{1,\dotsc,m\}$ and~$K\in(0,\infty)$  be arbitrary. Then,
	\[
		\inf_{u \gg 0} \sup_{x \in E_{j,K}} \frac{(Q_ju)(x)}{u(x)}
			 = \inf_{u \gg 0} \sup_{x \in E_{j,K}} \frac{(Q_j^\epsilon u)(x)}{u(x)},
	\]
	with the infimum taken over all functions of class $C^2$ which are strictly bounded away from 0 on the ellipsoid $E_{j,K} := \{x : \braket{(x-x_j),X_j(x-x_j)} < K \}$. In view of Lemma~\ref{lem:indep-p}, we may pick a strictly positive eigenfunction~$\psi^\epsilon$ for the eigenvalue~$\spb(A^\epsilon)$ of~$A^\epsilon$ which is of class~$C^2$.
	Hence,
	\begin{equation}\label{eq:intermediate-bound-a}
		\inf_{u \gg 0} \sup_{x \in E_{j,K}} \frac{(Q_ju)(x)}{u(x)}
			 \leq \sup_{x \in B(x_j,\epsilon^{1/2}K)} \frac{(Q_j^\epsilon (\psi^\epsilon)^{a})(x)}{(\psi^\epsilon(x))^a}.
	\end{equation}
	Now, by the chain rule and Young's inequality,
	{\small
	\begin{align*}
		\frac{(Q_j^\epsilon (\psi^\epsilon)^{a})(x)}{(\psi^\epsilon(x))^a} - \frac{a( A^\epsilon \psi^\epsilon)(x)} {(\psi^\epsilon(x))}
			&\leq \frac{a(Q_j^\epsilon \psi^\epsilon)(x)} {(\psi^\epsilon(x))} - \frac{a( A^\epsilon \psi^\epsilon)(x)} {(\psi^\epsilon(x))} +
		\frac{\epsilon a(a-1) |\nabla \psi^\epsilon(x)|^2} {(\psi^\epsilon(x))^2} - (1-a)w_j
			\\ &
			\leq  \epsilon^{-1} |W_0(x)- q_j(x)| + |W_1(x)-w_j|
			 \\ & \qquad
      + \frac{|F(x)-\ell_j(x)||\nabla \psi^\epsilon(x)|}{\psi^\epsilon(x)} - \frac{\epsilon a|1-a||\nabla \psi^\epsilon(x)|^2}{|\psi^\epsilon(x)|^2}
			- (1-a)w_j
			\\ &
			\leq  \epsilon^{-1} |W_0(x)- q_j(x)| + |W_1(x)-w_j|
			+ \frac{|F(x)-\ell_j(x)|^2}{4\epsilon a|1-a|} - (1-a)w_j.
	\end{align*}
	}
	Using the above in~\eqref{eq:intermediate-bound-a} and exploiting the regularity of~$F$, $W_0$ and $W_1$, we deduce that
	\begin{align*}
		\inf_{u \gg 0} \sup_{x \in E_{j,K}} \frac{(Q_ju)(x)}{u(x)}
		&\leq a \spb(A^\epsilon)
			+ C\Big(\epsilon^{3r-1} + \epsilon^r + \frac{1}{4a(1-a)}\epsilon^{4r-1}\Big)
			- (1-a) w_j
	\end{align*}
	for some constant~$C$ which is uniform in~$a$ and~$\epsilon$. Taking $\epsilon \to 0$ and then $a \to 1$ and using the Protter--Weinberger principle for the leading eigenvalue, we obtain
	\begin{align*}
		\spb \big(Q_j\!\!\upharpoonright_{E_{j,K}}\big) \leq \liminf_{\epsilon \to 0} \spb(\Lambda^\epsilon).
	\end{align*}
	Here,~``$\upharpoonright_{E_{j,K}}$'' denotes the restriction to~$E_{j,K}$ with a Dirichlet boundary condition.

	{
	Now note the following observation of~\cite{PW66}: if~$q$ is a function of class~$C^2$ which is strictly positive on~$E_{j,K}$ and vanishes at the boundary, then
	\[
		\sup_{x \in E_{j,K}} \frac{(Q_j q)(x)}{q(x)}
		\leq \spb \big(Q_j\!\!\upharpoonright_{E_{j,K}}\big).
	\]
	To see this, suppose that the inequality fails, let $\mu$ be strictly between the two members of the inequality, and derive a contradiction to the definite sign of the resolvent of~$Q_j\!\!\upharpoonright_{E_{j,K}}$ at $\mu$ (this last fact can be derived from the properties of the positivity-preserving semigroup). Using this observation with $q = q_{j,K}$ which is the shift $\phi_j - \Exp{-K}$ of the eigenfunction associated to the leading eigenvalue on the whole space and taking $K \to \infty$ gives
	\[
		 \spb Q_j \leq \liminf_{K\to\infty} \spb \big(Q_j\!\!\upharpoonright_{E_{j,K}}\big).
	\]
	It is precisely to have vanishing at the boundary with a simple constant shift that we chose $E_{j,K}$ in such a way that its boundary is a level set of~$\phi_j$.
	The desired lower bound holds.
	}

\begin{description}
	\item[Upper bound.] If Assumptions~\textnormal{(L0)},~\textnormal{(L1)},~\textnormal{(RB)} and~\textnormal{(ND)} are satisfied, then
	\[
		\limsup_{\epsilon \to 0} \spb(A^{\epsilon}) \leq \max_{j=1,\dotsc,m} \spb(Q_j).
	\]
\end{description}

Let $\chi :[0,\infty) \to [0,1]$ be a function of class~$C^2$ such that
$\chi(\rho) = 1$ for~$\rho \in [0,1]$, $\chi$~is strictly decreasing on~$(1,4)$ and $\chi(\rho) = 0$ for $\rho \in [4,\infty)$.
Note that the following quantity defined for $\beta \in [\tfrac 12 ,1)$ vanishes as $\beta \to 1$:
\[
	\gamma_\beta := \sup_{\chi(\rho) \geq \beta} |\nabla \chi(\rho)| + |\upDelta \chi(\rho)|.
\]
In order to focus on small neighbourhoods around the minima of~$W_0$, but which yet are large compared to the width of the eigenfunction~$\phi_j^\epsilon$ of~$Q_j^\epsilon$, we fix some
$
	r \in (\tfrac 13, \tfrac 12)
$
and set
\begin{equation*}
	\eta^\epsilon_j(x) := \chi(\epsilon^{-2r} \braket{x-x_j, X_j (x-x_j)})
\end{equation*}
for $j = 1, \dotsc, m$, and
\begin{equation*}
	\eta^\epsilon_0(x) := 1 - \sum_{j=1}^m \eta^\epsilon_j(x).
\end{equation*}
We consider~$\epsilon \in (0,\epsilon_0)$ with~$\epsilon_0 > 0$ small enough to guarantee $\supp \eta_j^\epsilon \cap \supp \eta_{j'}^\epsilon = \emptyset$ if $1 \leq j < j' \leq m$. We set
\[
	f_\beta^\epsilon(x) := \kappa_\beta^\epsilon \eta_0^\epsilon(x)+ \sum_{j=1}^m \eta^\epsilon_j(x) \phi_j^\epsilon(x),
\]
where
\[
	\kappa_\beta^\epsilon := \Exp{-\tfrac 12 \chi^{-1}(\beta)\epsilon^{2r-1}}.
\]
By the Protter--Weinberger principle,
\begin{equation}\label{eq:PW-ub-j}
	\spb A^\epsilon \leq \sup_{x \in \rr^N} \frac{(A^\epsilon f_\beta^\epsilon)(x)}{f_\beta^\epsilon(x)}
	 = \max\bigg\{\sup_{x : \eta_0^\epsilon(x) > 1-\beta} \frac{(A^\epsilon f_\beta^\epsilon)(x)}{f_\beta^\epsilon(x)},
	 \max_{j=1,\dotsc,m} \bigg\{\sup_{ x : \eta_j^\epsilon(x) \geq \beta } \frac{(A^\epsilon f_\beta^\epsilon)(x)}{f_\beta^\epsilon(x)} \bigg\}\bigg\}.
\end{equation}

Using Lemmas~\ref{lem:PW-out} and~\ref{lem:PW-in} below in~\eqref{eq:PW-ub-j} and taking~$\epsilon \to 0$ yields
\begin{align*}
	\limsup_{\epsilon \to 0} \spb A^\epsilon
	&  \leq \max_{j=1,\dotsc, m} \beta \spb Q_j + (\beta^{-1}-\beta)|w_j| +  C\gamma_\beta
\end{align*}
for some positive constant~$C$ independent of~$\beta$. Because $\beta \in [\tfrac 12 ,1)$ was arbitrary and both $\gamma_\beta \to 0$ and $\beta^{-1}-\beta \to 0$ as $\beta \to 1$, we conclude that
\begin{align*}
	\limsup_{\epsilon \to 0} \spb A^\epsilon
	&\leq
	 \max_{j=1,\dotsc, m} \spb Q_j.
\end{align*}
Before we state and prove Lemmas~\ref{lem:PW-out} and~\ref{lem:PW-in} to conclude the proof of the upper bound, let us give a collection of bounds which follow from the observation that
$\eta_j^\epsilon(x) \geq \beta$ if and only if~$\phi_j^\epsilon(x) \geq \kappa_\beta^\epsilon$.

\begin{lemma}\label{lem:beta-bounds}
	There exists a constant~$C$ with the following property:
	\begin{enumerate}
		\item[i.] if $\eta_0^\epsilon(x) > 1 - \beta$, then
		$
			0 < \frac{\phi_j^\epsilon(x)}{f_\beta^\epsilon(x)}
			< \frac{\kappa_\beta^\epsilon}{f_\beta^\epsilon(x)} < \frac 1{1-\beta}
		$
		and
		$
			\epsilon^{r}|\nabla \eta_j^\epsilon (x)| + \epsilon^{2r}|\upDelta \eta_j^\epsilon (x)| \leq C
		$
		for each~$j \in \{1,\dotsc,m\}$;
		\item[ii.] if $\eta_j^\epsilon(x) \geq \beta$, then
		$
			 0 \leq \frac{\kappa_\beta^\epsilon}{f_\beta^\epsilon(x)} \leq 1 \leq \frac{\phi_j^\epsilon(x)}{f_\beta^\epsilon(x)} \leq \frac{1}{\beta}
		$
		and
		$
			 \epsilon^{r}|\nabla \eta_j^\epsilon (x)| + \epsilon^{2r}|\upDelta \eta_j^\epsilon (x)| \leq C \gamma_\beta
		$.
	\end{enumerate}
\end{lemma}

\begin{lemma}\label{lem:PW-out}
	There exists strictly positive constants~$C$ and~$\delta$ such that
	\[
		\sup_{x : \eta_0^\epsilon(x) > 1-\beta} \frac{(A^\epsilon f_\beta^\epsilon)(x)}{f_\beta^\epsilon(x)} \leq - (1-\beta) \delta \epsilon^{2r-1} + C
	\]
	for all~$\epsilon \in (0,\epsilon_0)$ and all~$\beta \in [\tfrac 12,1)$.
\end{lemma}

\begin{proof}
	Let $x$ such that~$\eta_0^\epsilon(x) > 1-\beta$ be arbitrary. Throughout the proof, the big $O$ notation refers to constants that are uniform in~$x$,~$\epsilon$ and~$\beta$.
	We compute
	\begin{equation}\label{eq:nabla-f}
		\nabla f_\beta^\epsilon(x) =  \kappa_\beta^\epsilon \nabla \eta_0^\epsilon(x) +
		\sum_{j=1}^m \phi_j^\epsilon(x) \nabla \eta_j^\epsilon(x) - \epsilon^{-1}\eta_j^\epsilon(x)\phi_j^\epsilon(x)X_j (x-x_j),
	\end{equation}
	and
	\begin{multline}\label{eq:upDelta-f}
		\upDelta f_\beta^\epsilon(x) =  \kappa_\beta^\epsilon \upDelta \eta_0^\epsilon(x) +
		\sum_{j=1}^m \phi_j^\epsilon(x)\upDelta \eta_j^\epsilon(x) + 2\epsilon^{-1}\braket{ \phi_j^\epsilon(x) X_j (x-x_j), \nabla \eta_j^\epsilon(x)}  \\ -  \epsilon^{-1}\eta_j^\epsilon(x)\phi_j^\epsilon(x)\tr X_j + \epsilon^{-2} \eta_j^\epsilon(x) \phi_j^\epsilon(x) |X_j (x-x_j)|^2.
	\end{multline}
	Hence, using Lemma~\ref{lem:beta-bounds}.i and the fact that $|F(x)| = O(\epsilon^r)$ on~$\supp \nabla \eta_0^\epsilon$ and~$\supp \nabla \eta_j^\epsilon$,
	\begin{equation*}
		\frac{\epsilon \upDelta f_\beta^\epsilon(x) + \braket{F(x),\nabla f_\beta^\epsilon(x)}}{f_\beta^\epsilon}
		 =
		 \sum_{j=1}^m\frac{\phi_j^\epsilon(x)\eta_j^\epsilon(x)}{\epsilon f_\beta^\epsilon(x)} \braket{x-x_j, X_j^2(x-x_j) - X_j F(x)}  + O(1).
	\end{equation*}
	Using $|F - \ell_j| = O(\epsilon^{2r})$ on~$\supp \eta_j^\epsilon$ and then the \textsc{are}, we obtain
	\begin{align*}
		\frac{\epsilon \upDelta f_\beta^\epsilon(x) + \braket{F(x),\nabla f_\beta^\epsilon(x)}}{f_\beta^\epsilon}
		&= \epsilon^{-1} \sum_{j=1}^m\frac{\phi_j^\epsilon(x)\eta_j^\epsilon(x)}{f_\beta^\epsilon(x)} q_j(x) + O(1).
	\end{align*}
	Using Lemma~\ref{lem:beta-bounds}.i again,
	\begin{align}\label{eq:deriv-terms}
		\frac{\epsilon \upDelta f_\beta^\epsilon(x) + \braket{F(x),\nabla f_\beta^\epsilon(x)}}{f_\beta^\epsilon}
		&\leq \epsilon^{-1} \beta \sum_{j=1}^m \one_{\supp \eta_j^\epsilon} q_j(x) + O(1).
	\end{align}
	Substracting
	\begin{align*}
		\epsilon^{-1} W_0 + W_1 \geq \epsilon^{-1}\beta \Big(\sum_{j=1}^m \one_{\supp_{\eta_j^\epsilon}}W_0\Big) + (1-\beta) W_0 + O(1)
	\end{align*}
	---\,we have used~\textnormal{(L0)} and~$\alpha \in \cA$ to obtain $\beta\epsilon^{-1}(1 - \sum_j \one_{\supp \eta_j^\epsilon})W_0 + W_1 \geq O(1)$\,---\,from~\eqref{eq:deriv-terms}, we obtain
	\begin{align*}
		\frac{(A^\epsilon f_\beta^\epsilon)(x)}{f_\beta^\epsilon(x)}
		&\leq \epsilon^{-1} \sum_{j=1}^m\beta \one_{\supp_{\eta_j^\epsilon}}(x) (q_j(x) - W_0(x))
			 - \epsilon^{-1} (1-\beta) W_0(x) + O(1).
	\end{align*}
	Now, because $|W_0 - q_j| = O(\epsilon^{3r})$ on~$\supp \eta_j^\epsilon$, we have
	\begin{align*}
		\frac{(A^\epsilon f_\beta^\epsilon)(x)}{f_\beta^\epsilon(x)}
		&\leq \epsilon^{-1} \beta O(\epsilon^{3r}) - \epsilon^{-1}(1-\beta)W_0(x) + O(1).
	\end{align*}
	Because~$W_0 \geq 0$ with nondegenerate zeroes precisely in~$\{x_j\}_{j=1}^m$, and because the set $\{x : \eta_j^\epsilon(x) < \beta \}$ excludes a ball of radius of order~$\epsilon^r$ around~$x_j$, there exists a strictly positive constant~$\delta > 0$ such that $W_0(x) > \delta\epsilon^{2r}$ for all~$x$ such that $\eta_0^\epsilon(x) > 1 - \beta$.
\end{proof}

\begin{lemma}\label{lem:PW-in}
	There exists a positive constant~$C$ such that
	\begin{equation*}
		\sup_{\{\eta_j^\epsilon(x) \geq \beta\}} \frac{(A^\epsilon f_\beta^\epsilon)(x)}{\ f_\beta^\epsilon(x)}   \leq \beta \spb Q_j^\epsilon + (\beta^{-1}-\beta)|w_j| +  C\big(\gamma_\beta (1+\beta^{-1})(1+\epsilon^{1-2r}) + \epsilon^{3r-1} + \epsilon^r\big).
	\end{equation*}
	for all~$\epsilon \in (0, \epsilon_0)$ and all~$\beta \in [\tfrac 12,1)$.
\end{lemma}

\begin{proof}
	Let~$x$ such that $\eta_j^\epsilon(x) \geq \beta$ be arbitrary. In particular, $|x-x_j| = O(\epsilon^r)$. Throughout the proof, the big $O$ notation refers to constants that are independent of~$x$,~$\epsilon$ and~$\beta$. By~\eqref{eq:nabla-f},~\eqref{eq:upDelta-f}, Lemma~\ref{lem:beta-bounds}.ii and the fact that $|F| = O(\epsilon^r)$,
	\begin{align}\label{eq:rid-O-gamma}
		\bigg|\frac{(A^\epsilon f_\beta^\epsilon)(x)}{f_\beta^\epsilon(x)} - \frac{\eta_j^\epsilon(x)(A^\epsilon \phi_j^\epsilon)(x)}{f_\beta^\epsilon(x)} \bigg| \leq C\gamma_\beta (1 + \beta^{-1}) (\epsilon^{1-2r} + 1).
	\end{align}
	Now, using $|F(x)-\ell_j(x)| = O(\epsilon^{2r})$,~$\nabla \phi_j^\epsilon(x) / \phi_j^\epsilon(x) = \epsilon^{-1}O(\epsilon^r)$,~$|W_0(x) - q_j(x)| = O(\epsilon^{3r})$ and~$|W_1(x)-w_j(x)| = O(\epsilon^r)$ for $x \in \supp \eta_j^\epsilon$,
	\begin{align*}
		\frac{\eta_j^\epsilon(x)(A^\epsilon \phi_j^\epsilon)(x)}{f_\beta^\epsilon(x)}
		&= \frac{\eta_j^\epsilon(x)\phi_j^\epsilon(x)}{f_\beta^\epsilon(x)}
		\bigg(\frac{([Q_j^\epsilon - |w_j|] \phi_j^\epsilon)(x)}{\phi_j^\epsilon(x)} + |w_j| + \epsilon^{-1}O(\epsilon^{3r})+O(\epsilon^r )\bigg).
	\end{align*}
Because $\phi_j^\epsilon$ is an eigenvector of $[Q_j^\epsilon - |w_j|]$ with eigenvalue $\spb Q_j - |w_j| \leq 0$ and because the prefactor on the right-hand side lies in the interval $[\beta,\beta^{-1}]$ by Lemma~\ref{lem:beta-bounds}.ii, we have
\begin{align}\label{eq:rid-o-in-eps}
	\frac{\eta_j^\epsilon(x)(A^\epsilon \phi_j^\epsilon)(x)}{f_\beta^\epsilon(x)} &\leq
	\beta \spb Q_j + (\beta^{-1}-\beta) |w_j|  + C(\epsilon^{3r-1} + \epsilon^r).
\end{align}
Combining~\eqref{eq:rid-O-gamma} and~\eqref{eq:rid-o-in-eps} and using the fact that $\eta_j^\epsilon(x) \geq \beta$ implies $|F(x)| = O(\epsilon^r)$, we conclude that a bound of the proposed form indeed holds.
\end{proof}

\appendix

  \section{Properties of the deformed generators}
  \label{app:semig}

  In this appendix, we collect some results from the theory of semigroups applied to partial differential equations involving elliptic operators of the form
  \begin{equation}
    \Lambda^{\epsilon,\alpha} := \epsilon \upDelta +  \braket{-\nabla V + (1- 2\alpha) b,\nabla } - \tfrac{\alpha(1-\alpha)}{\epsilon} |b|^2  + \tfrac {\alpha}{\epsilon} \braket{b, \nabla V} - \alpha\operatorname{div} b,
  \end{equation}
  which play a key role in Sections~\ref{sec:large-t},~\ref{sec:small-e} and~\ref{app:semiclass} of the paper, similarly as in Appendix~A of~\cite{BD15} (the case where $b$ is bounded). They are deformations of the infinitesimal generator of the semigroup associated to~\eqref{eq:SDE}.

  We use technical results from the article~\cite{MPRS05}, Chapter~1 of~\cite{Lan} and Chapters~A-I,~C-IV and~B-IV of~\cite{AGG+}. Throughout this section, whenever we refer to~$V$ and~$b$, we assume that~\textnormal{(L0)},~\textnormal{(L1)} and~\textnormal{(RB)} hold. Also, we write~$\mathrm{L}^p(\rr^N)$ for~$\mathrm{L}^p(\rr^N,\d\vol)$, and similarly for the Sobolev spaces. For the spaces~$C(\rr^N)$ of continuous functions and~$C^k(\rr^N)$ of~$k$-times differentiable functions, the subscript ``0'' is used for ``vanishing at infinity''; ``$\mathrm{c}$'', for ``compactly supported''.

  For $p \in (1,\infty)$, a straightforward computation shows that
  \[
    \Exp{-(p\epsilon)^{-1}V}\Lambda^{\alpha,\epsilon}(\Exp{(p\epsilon)^{-1}V}f) =  \epsilon\upDelta f + \braket{F_p,\nabla f} - \Omega_p f
  \]
  for all $f \in C^2_\mathrm{c}(\rr^N)$, where
  \[
    F_p := (\tfrac 2{p} - 1)\nabla V + (1-2\alpha) b,
  \]
  \[
    \Omega_p := \tfrac 1\epsilon
    W_0
    - \tfrac 1{p} \upDelta V + \alpha\operatorname{div} b
  \]
  and
  \[
     W_0 := \tfrac 1p (1- \tfrac 1p) |\nabla V|^2  - \tfrac{1 -2\alpha +\alpha p }{p}\braket{b,\nabla V} + \alpha(1-\alpha) |b|^2.
  \]
  For technical reasons, we need to restrict our attention to a certain $\alpha$-dependent set of powers~$p$. We introduce an admissibility condition for the pair~$(\alpha,p)$.
  \begin{definition}
   {The pair~$(\alpha,p) \in \rr \times (1,\infty)$ is said to be~\emph{admissible} if there exists $\ell\in(0,1)$
   such that
   \[
  		\inf \left\{ \ell \tfrac {1}{p} \big(1 - \tfrac 1p \big) |\nabla V(x)|^2  - \tfrac{1 -2\alpha +\alpha p }{p}\braket{b(x),\nabla V(x)} + \alpha(1-\alpha) |b(x)|^2 : x \in \rr^d \right\}> -\infty.
   \]
	 }
  \end{definition}
  The next lemma\,---\,whose proof follows from straightforward applications of~\textnormal{(RB)} and the Cauchy--Schwarz inequality\,---\,gives concrete sufficient conditions for admissibility. These conditions are illustrated in Figure~\ref{fig:alpha-and-p}.
  \begin{lemma}\label{lem:link-I-Abb}
  	Let~$k_b \in [0,\tfrac 12)$ and $h_b \in [0,\infty)$ be as in assumption~\textnormal{(RB)}. If \[1-2\alpha+\alpha p \geq 0\] and either
  	\begin{itemize}
  		\item[i.] we have $\alpha (1 - \alpha)  \geq 0$ and $1- p^{-1} -(1 - 2\alpha + \alpha p)k_b  > 0$ or
  		\item[ii.] we have $\alpha (1 - \alpha)  < 0$ and $1- p^{-1} - (1 - 2\alpha + \alpha p)k_b - p\alpha(\alpha - 1)h_b > 0$,
  	\end{itemize}
  	then the pair~$(\alpha,p)$ is admissible. In particular, if $p$ is fixed in the interval~$(\tfrac{1}{1-k_b},\tfrac{1}{k_b})$, then the pair~$(\alpha,p)$ is admissible for all~$\alpha$ in an open interval containing~$[0,1]$.
  \end{lemma}

  \begin{figure}
    \begin{center}
		\includegraphics{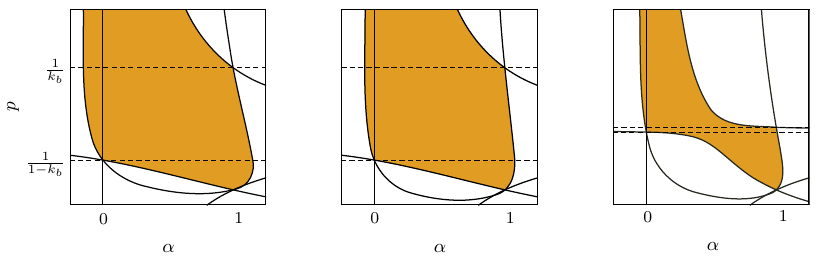}
  	\end{center}
    \caption{The orange region enclosed in the solid contours is the set of values allowed in Lemma~\ref{lem:link-I-Abb} computed for~$(k_b,h_b) = (0.33,0.75), (0.33,1.5)$ and~$(0.49,1.5)$\,---\,from left to right.
  	}
    \label{fig:alpha-and-p}
  \end{figure}

  Until further notice, we fix $\alpha$, $p$ and~$\ell$ as in the admissibility condition.
  By Assumption~\textnormal{(L0)} and the fact that~$b$ is globally Lipschitz, there exist~$c^0_p$ and $\theta \in (0,1)$ such that
  \begin{equation}\label{eq:MH5}
    |\operatorname{div} F_p| \leq \theta  \Big((1-\ell)\tfrac {1}{\epsilon p} (1- \tfrac 1p)|\nabla V|^2 + c^0_p\Big).
  \end{equation}
  Set
  \[
    U_p := (1-\ell) \tfrac 1{\epsilon p} (1- \tfrac 1p)|\nabla V|^2 + c^0_p.
  \]
  Using the same properties again, we may pick~$\kappa$ such that
  \begin{equation}\label{eq:MH4}
    |F_p| \leq \kappa U_p^{\frac 12}.
  \end{equation}
  Using Assumption~\textnormal{(L0)} and the admissibility condition, we can pick positive constants~$c_p$ and~$c^1_p$ such that
  \begin{equation}\label{eq:MH3}
    U_p \leq \Omega_p + c_p \leq c^1_p U_p.
  \end{equation}

  \begin{lemma}\label{lem:MPRS-as-one}
  	Suppose that the pair~$(\alpha,p)$ is admissible. Then, the operator
  	\[
  		A_p := \epsilon\upDelta + \braket{F_p,\nabla} - \Omega_p - c_p
  	\]
  	with domain
  	\[
  		\cD^q := \{f \in \mathrm{W}^{2,q}(\rr^N) : U_p f \in \mathrm{L}^q(\rr^N)\}
  	\]
  	is closed as an unbounded operator on~$\mathrm{L}^q(\rr^N)$ and generates an analytic, compact, positivity-preserving semigroup on~$\mathrm{L}^q(\rr^N)$ for all~$q \in (1,\infty)$. With domain
  	\[
  	  \cD^\infty := \{f \in C_0(\rr^N) : f \in \mathrm{W}_{\textnormal{loc}}^{2,q} \textnormal { for all } q \in (1,\infty) \textnormal{ and } \upDelta f, U_p f \in C_0(\rr^N)\},
  	\]
  	it is closed as an unbounded operator on~$C_0(\rr^N)$ and generates an analytic, compact, positivity-preserving semigroup on~$C_0(\rr^N)$.
  \end{lemma}

  \begin{proof}
  	For any real number~$r > 0$,
  	by~\textnormal{(L0)} and Cauchy's inequality, there exists~$C_{p,r} > 0$ such that
    \begin{equation}\label{eq:MH2}
      |\nabla U_p| \leq 16 r U_p^{\frac 32} + C_{p,r}.
    \end{equation}
  	The bounds~\eqref{eq:MH5}--\eqref{eq:MH2} precisely give hypotheses~(H2)--(H5) of~\cite{MPRS05}. Therefore, Theorem~3.4 in~\cite{MPRS05} gives that $(A_p,\cD^q)$ generates a holomorphic positivity-preserving semigroup on~$\mathrm{L}^q(\rr^N)$, and Theorem~4.4 in~\cite{MPRS05} gives that $(A_p,\cD^\infty)$ generates a holomorphic positivity-preserving semigroup on~$C_0(\rr^N)$. Compactness follows from Proposition~6.4 in~\cite{MPRS05}.
  \end{proof}

  \begin{lemma}\label{lem:irred-coinc}
  	The semigroups in Lemma~\ref{lem:MPRS-as-one} coincide on the intersection of their spaces of definition and are all irreducible (positivity improving) on their respective spaces.
  \end{lemma}

  \begin{proof}
  	The first part of the lemma is proved as Lemma~4.3 in~\cite{MPRS05}. The second part follows from the strong maximal principle; see \textit{Step 6} in the proof of Lemma~A.1 in~\cite{BD15}.
  \end{proof}

  \begin{lemma}\label{lem:first-time-evec}
  	Suppose that the pair~$(\alpha,p)$ is admissible and let
  	\[
  		s_p := \sup \{\Re z : z \in \sp(A_p,\cD^p)\}.
  	\]
  	Then,~$s_p$ is a simple isolated eigenvalue and there exist a strictly positive
  	vector~$\tilde{\psi}_p \in \cD^p$ and a strictly positive functional~$\tilde{u}_p$ on~$\mathrm{L}^p(\rr^N)$ such that
  	\begin{equation}
  	\label{eq:convergence-Lp}
  		\lim_{t \to \infty} \big \| \Exp{-t s_p}\Exp{t A_p} f - \tilde{\psi}_p (\tilde{u}_p, f) \big \|_{p} = 0
  	\end{equation}
  	for all~$f \in \mathrm{L}^p(\rr^N)$.
  \end{lemma}

  \begin{proof}
  	This is a well-established consequence of irreducibility, compactness and preservation of positivity; see Theorem 2.1 and Remark 2.2(e) in~\cite[Ch.\,{C-IV}]{AGG+}.
  \end{proof}

  \begin{lemma}\label{lem:indep-p}
  	For all~$f \in C_\mathrm{c}(\rr^N)$, the convergence expressed in~\eqref{eq:convergence-Lp} holds in the norm~$\|\,\cdot\,\|_q$ for all~$q \in (1,\infty]$. Moreover, the vector~$\tilde{\psi}_p$ has a representative which is strictly positive, twice continuously differentiable, vanishes at infinity and belongs to~$\mathrm{L}^q(\rr^N)$ for all~$q \in (1,\infty]$.
  	If~$(\alpha,p)$ and~$(\alpha,p')$ are both admissible, then~$s_p + c_p$ conicides with~$s_{p'} + c_{p'}$.
  \end{lemma}

  \begin{proof}
  	By the same argument giving Lemma~\ref{lem:first-time-evec}, there exist a real number~$\bar{s}_{p,q}$, a strictly positive vector~$\bar{\psi}_{p,q} \in \cD^q$ and a strictly positive functional~$\bar{u}_{p,q}$ on~$\mathrm{L}^q(\rr^N)$ or~$C_0(\rr^N)$ such that
  	\begin{equation}
  		\label{eq:convergence-Lq}
  		\lim_{t \to \infty} \big \| \Exp{-t \bar{s}_{p,q}}\Exp{t A_{p}} f - \bar{\psi}_{p,q} (\bar{u}_{p,q}, f) \big \|_{q} = 0
  	\end{equation}
  	for all~$f$ in~$\mathrm{L}^q(\rr^N)$ or~$C_0(\rr^N)$; see Corollary~2.2 in~\cite[Ch.\,{B-IV}]{AGG+} for~$q=\infty$. Taking a common nonnegative~$f \in C_\mathrm{c}(\rr^N)\setminus \{ 0 \}$ in both~\eqref{eq:convergence-Lp} and~\eqref{eq:convergence-Lq}
  	and using Lemma~\ref{lem:irred-coinc} gives~$\bar{s}_{p,q} = {s}_p$ and $\bar{\psi}_{p,q} \propto \tilde{\psi}_p$.
  	Because~$A_p + c_p$ and~$A_{p'} + c_{p'}$ are related by a conjugation which preserves~$C_\mathrm{c}(\rr^N)$, a similar argument also yields that e.g.~$\bar{s}_{p,\infty} + c_{p}$ coincides with~$\bar{s}_{p',\infty}+ c_{p'}$.

  	Note that~$\tilde{\psi}_2$ is in a H\"older space~$C^{1,\beta}(\rr^N)$ with $\beta \in (0,1)$ by a Sobolev embedding. The approximation method for inferring that~$\tilde{\psi}_2$ belongs to~$C^{2,\beta}(\rr^N)$ via classical interior Schauder estimates and the maximum principle is carried out in~\cite[\S{1.8}]{Lan}.
  \end{proof}

  It is proved as part of Theorem~7.4 in~\cite{MPRS05} that the isometry $f \mapsto \Exp{(p\epsilon)^{-1}V}f$ between the Banach spaces~$\mathrm{L}^p(\rr^N)$ and~$\mathrm{L}^p(\rr^N,\d\mu_0^\epsilon)$ used to introduce~$A_p$ maps the domain~$\cD^p$ to $\mathrm{W}^{2,p}(\rr^N,\d\mu_0^\epsilon)$.
  Hence, it follows immediately from Lemmas~\ref{lem:MPRS-as-one} and~\ref{lem:irred-coinc} that~$(\Lambda^{\alpha,\epsilon},\mathrm{W}^{2,p}(\rr^N,\d\mu_0^\epsilon))$ is the generator of an analytic semigroup which is compact and irreducible, provided that~$(\alpha,p)$ is admissible.
  Also, by Lemmas~\ref{lem:first-time-evec} and~\ref{lem:indep-p},
  \[
    e^\epsilon(\alpha) := \sup\{ \Re z : z \in \sp (\Lambda^{\alpha,\epsilon},\mathrm{W}^{2,p}(\rr^N,\d\mu_0^\epsilon))\}
  \]
  is indeed independent of~$p$ and admits an eigenvector with the properties stated in the proposition below.

  \begin{proposition}\label{prop:evec}
    Let the pair~$(\alpha,p)$ be admissible. Then,~$e^\epsilon(\alpha)$ is a simple isolated eigenvalue of~$(\Lambda^{\alpha,\epsilon},\mathrm{W}^{2,p}(\rr^N,\d\mu_0^\epsilon))$
    and there exists a strictly positive associated eigenfunction $\psi^{\alpha,\epsilon}  \in C^2(\rr^N) \cap \mathrm{W}^{2,p}(\rr^N,\d\mu_0^\epsilon)$ and a strictly positive linear functional~$u^{\alpha,\epsilon}$ on~$\mathrm{L}^{p}(\rr^N,\d\mu_0^\epsilon)$ such that
  	\begin{equation*}
  		\lim_{t \to \infty} \big \| \Exp{-t e^\epsilon(\alpha)}\Exp{t \Lambda^{\alpha,\epsilon}} f - \psi^{\alpha,\epsilon} (u^{\alpha,\epsilon} ,f )_{\mu_0^\epsilon} \big \|_{\mathrm{L}^p(\rr^N,\d\mu_0^\epsilon)} = 0
  	\end{equation*}
  	for all~$f \in \mathrm{L}^{p}(\rr^N,\d\mu_0^\epsilon)$.
  \end{proposition}

	\section{Standard probabilistic consequences of our assumptions}
	\label{app:cons}

	Under our assumptions, existence and uniqueness of the solutions to the \textsc{sde}~\eqref{eq:SDE} is standard; we refer the reader to~\cite[\S{3.3}]{Has}. Let us only mention that the following consequence of~\textnormal{(L1)} plays a key role in the proof. It is also used throughout the paper.

	\begin{lemma}\label{lem:dissipation-with-sup}
		Suppose that Assumption~\textnormal{(L1)} holds. Then,
		\[
			\pp\bigg\{ \sup_{s \in [0,t]} |X^{x,\epsilon}_s| \geq R \bigg\}
				\leq  \frac{\braket{x,H_b x} + 2K_b + 2\epsilon \tr H_b }{R^2 \inf \sp H_b } \Exp{ t}
		\]
		for all~$t \geq 0$, $x \in \rr^N$ and~$R > 0$.
	\end{lemma}

	\begin{proof}
		Using~\textnormal{(L1)}, follow the first steps of the proof of Theorem~3.5 in~\cite[Ch.\,3]{Has} with the nonnegative function~$x \mapsto \braket{x,H_b x} + 2K_b + 2\epsilon \tr H_b$ and $c=1$.
	\end{proof}

	Since the diffusion matrix is nondegenerate, existence and uniqueness of the stationary measure~$\lambda^\epsilon_\textnormal{inv}$  for~\eqref{eq:SDE} can be derived if one controls the expected hitting time of a large enough ball, uniformly on compact sets of initial conditions~\cite[\S{4.4}]{Has}.
	The following estimate is a key step in controlling these hitting times and is also used in the main body of the article. It is again a consequence of~\textnormal{(L1)}.

	\begin{lemma}\label{lem:dissipation-without-sup}
		Let~$H_b$ be as in Assumption~\textnormal{(L1)}. Then, for all~$\epsilon_0 > 0$, there exist positive constants~$c$ and~$C$
		such that
		\begin{equation}
			\label{eq:dissipation-without-sup}
			0 < \inf\sp H_b \ \ee  | X^{x,\epsilon}_t|^2 \leq \ee  \braket{X^{x,\epsilon}_t,H_b X^{x,\epsilon}_t}
			\leq \Exp{-ct}\braket{x,H_b x} + C
		\end{equation}
		for all $\epsilon \in (0,\epsilon_0)$ and $t \geq 0$ and almost all~$x \in \rr^N$.
	\end{lemma}

	\begin{proof}
		The first two inequalities in~\eqref{eq:dissipation-without-sup} are immediate from the fact that~$H_b$ is positive definite. Let~$f : x \mapsto \braket{x,H_b x}$.
		By Kolmogorov's backwards equation\,---\,see e.g.\ Lemma~3.3 in~\cite[Ch.\,3]{Has}\,---, $\partial_t \ee f_R(X^{x,\epsilon}_t) = \Lambda^{\epsilon,0} \ee f_R(X^{x,\epsilon}_t)$ for any approximation~$f_R \in C^2_\mathrm{c}(\rr^N)$ of~$f$.
		We showed in Appendix~\ref{app:semig} that $(\Lambda^{\epsilon,0},W^{2,2}(\rr^N,\d\mu_0^\epsilon))$ generates a strongly continuous semigroup of bounded linear operators on~$\mathrm{L}^2(\rr^N,\d\mu_0^\epsilon)$. Hence, using basic semigroup properties\,---\,see e.g.\
		Proposition~1.6.ii and Theorem~1.7 in~\cite[Ch.\,{A-I}]{AGG+}\,---,
		we find
		\begin{align*}
			\partial_t \ee f(X^{x,\epsilon}_t)
			&= \ee \big[\epsilon \upDelta f(X^{x,\epsilon}_t) + \braket{-\nabla V(X^{x,\epsilon}_t) +b(X^{x,\epsilon}_t), \nabla f(X^{x,\epsilon}_t)}\big]
		\end{align*}
		for almost all~$x$ by an approximation argument.
		Then, by~\textnormal{(L1)},
		\begin{align*}
			\partial_t \ee \braket{X^x_t,H_b X^x_t}
			&\leq 2\epsilon\tr H_b  - 2 \braket{X^x_t, X^x_t} + 2K_b
			\leq 2\epsilon\tr H_b + 2K_b - 2\| H_b\|^{-1} \ee \braket{X^x_t,H_b X^x_t}
		\end{align*}
		for almost all~$x$ and the last inequality in~\eqref{eq:dissipation-without-sup} follows from Gr\"onwall's lemma.
	\end{proof}

  \begin{lemma}\label{lem:pre-ub-ic-RN}
  	The measure~$\lambda_\textnormal{inv}^\epsilon$ is of the form
    \[
      \lambda_\textnormal{inv}^\epsilon (\d x) = \Exp{-(2\epsilon)^{-1} V(x)} \varphi^{\epsilon}(x) \d x
    \]
    for some strictly positive function~$\varphi^{\epsilon} \in C^2_0(\rr^N) \cap \mathrm{L}^{2}(\rr^N)$.
  \end{lemma}

  \begin{proof}
  	Consider the operator~$(A_2,\cD^2)$ introduced in Appendix~\ref{app:semig} the case $\alpha = 1$, that is
  	\[
  		A_2 = \epsilon\upDelta - \braket{ b, \nabla} - \tfrac 1{4\epsilon}|\nabla V|^2 + \tfrac{1}{2\epsilon}\braket{b,\nabla V} + \tfrac 12 \upDelta V - \operatorname{div} b - c_2.
  	\]
  	One can show that its adjoint has domain~$\cD^2$ and is given by the formula
  	\[
  		A_2^* = \epsilon\upDelta + \braket{ b, \nabla} - \tfrac 1{4\epsilon}|\nabla V|^2 + \tfrac{1}{2\epsilon}\braket{b,\nabla V} + \tfrac 12 \upDelta V - c_2.
  	\]
  	Note that~$A_2^*$ just as well satisfies (H1)--(H5) in~\cite{MPRS05} and thus generates a semigroup with the same properties.
  	Note that~$\Exp{-(2\epsilon)^{-1}V}$ is a strictly positive eigenvector of~$(A_2^*,\cD^2)$ with eigenvalue~$-c_2$.
  	But it is easy to show by contradiction that~$\spb(A_2^*,\cD^2)$ is the \textit{only} eigenvalue of~$A_2$ admitting a strictly positive eigenvector.
  	Hence, we have $\spb(A_2,\cD^2) = \spb(A_2^*,\cD^2) = - c_2$.

  	Therefore, there exists a strictly positive function~$\phi^\epsilon \in C^2_0(\rr^N) \cap \mathrm{L}^{2}(\rr^N)$ such that $A_2 \phi^\epsilon = -c_2\phi^\epsilon$.
  	Then, $\rho^\epsilon := \Exp{-(2\epsilon)^{-1}V}\phi^\epsilon$ satisfies the stationary Fokker--Planck equation
  	\[
  		(\epsilon \upDelta + \braket{\nabla V - b, \nabla} + \upDelta V - \operatorname{div} b)\rho^\epsilon = 0,
  	\]
  	to which the density of the invariant measure~$\lambda_{\textnormal{inv}}^\epsilon$ is\,---\,up to normalisation\,---\,the unique bounded solution; see e.g.\ Lemma~4.16 in~\cite[Ch.\,{4}]{Has}.
  \end{proof}

  \begin{lemma}\label{lem:ub-ic-RN}
  	For all~$\beta \in (0,2)$, $(\d \lambda_{\textnormal{inv}}^\epsilon/\d\mu_0^\epsilon)^\beta \in \mathrm{L}^1(\rr^N,\d\mu_0^\epsilon).$
  \end{lemma}

  \begin{proof}
  	Set $r := 2\beta ^{-1}$ and let~$\varphi^\epsilon$ be as in Lemma~\ref{lem:pre-ub-ic-RN}. Then, by H\"older's inequality,
  	\begin{align*}
  		\int_{\rr^N} \Big|\frac{\d \lambda_{\textnormal{inv}}^\epsilon}{\d\mu_0^\epsilon}\Big|^\beta  \d\mu_0^\epsilon
  		& = \int_{\rr^N}  \big|\varphi^\epsilon\big|^\beta   \Exp{\beta (2\epsilon)^{-1} V} \Exp{-\epsilon^{-1}V} \d\vol \\
  		& \leq \bigg(\int_{\rr^N} \big|\varphi^{\epsilon}\big|^{\beta r}\d\vol \bigg)^{\frac 1r}
  		 \bigg(\int_{\rr^N} \Exp{-\epsilon^{-1}(1-r^{-1})^{-1}(1 - \frac 12 \beta )V}\d\vol \bigg)^{1 - \frac 1r}.
  	\end{align*}
  	Since $\beta r = 2$, the first integral is a power of the $\mathrm{L}^2(\rr^N)$-norm of~$\varphi^\epsilon$, which is finite by Lemma~\ref{lem:pre-ub-ic-RN}. The second integral is finite because $(1-r^{-1})(1- \tfrac 12 \beta )$ is strictly positive and~$V$ satisfies~\textnormal{(L0)}.
  \end{proof}

	Recall that $\pi_s : \cCt \to \rr^N$ is evaluation map $\gamma \mapsto \gamma(s)$ and that {time reversal} is the unique involution $\Theta_t : \cCt \to \cCt$ determined by the relation
	$
		\pi_s \circ \Theta_t = \pi_{t-s}
		.
	$
	We have used the identity~\eqref{eq:change-ic} in the proof of Proposition~\ref{prop:RN}, i.e.\ to give a more explicit epression for the canonical entropy production functional. We state and prove it as a lemma.

	\begin{lemma}
	\label{lem:change-ic}
		Under Assumption~\textnormal{(L0)}, if~$\lambda$ and the Lebesgue measure are mutually absolutely continuous, then~$\mathcal{Q}_t^{\lambda,\epsilon}$ and~$\mathcal{Q}_t^{\lambda,\epsilon} \circ \Theta_t$ are mutually absolutely continuous and
		\[
			\log \frac{\d\mathcal{Q}_t^{\lambda,\epsilon}}{\d(\mathcal{Q}_t^{\lambda,\epsilon} \circ \Theta_t)}(\gamma) = \log \frac{\d\lambda}{\d\mu_0^\epsilon}(\gamma(0))
			-
			\log \frac{\d\lambda}{\d\mu_0^\epsilon}(\gamma(t))
		\]
		for $\mathcal{Q}_t^{\lambda,\epsilon}$-almost all~$\gamma \in \cCt$.
	\end{lemma}

	\begin{proof}
		Let $\Gamma$ be a measurable subset of~$\cCt$. Using~\eqref{eq:disint},
		\begin{align*}
			\mathcal{Q}^{\lambda,\epsilon }_t(\Gamma)
				&= \int_{\rr^N} \int_{\cCt} \one_{\Gamma}(\gamma)	\one_{\pi_0^{-1}\{x\}}(\gamma) \, \mathcal{Q}^{x,\epsilon}_t(\d\gamma) \frac
				{\d\lambda}{\d\mu_0^{\epsilon}}(x) \mu_0^\epsilon(\d x) \\
				&= \int_{\rr^N} \int_{\cCt} \one_{\Gamma}(\gamma) \frac
				{\d\lambda}{\d\mu_0^{\epsilon}}(\pi_0 \gamma)	\one_{\pi_0^{-1}\{x\}}(\gamma)  \, \mathcal{Q}^{x, \epsilon}_t(\d\gamma) \,  \mu_0^\epsilon(\d x) \\
				&= \int_{\cCt} \one_{\Gamma}(\gamma) \, \frac
				{\d\lambda}{\d\mu_0^\epsilon}(\pi_0 \gamma)   \mathcal{Q}^{ \mu_0,\epsilon}_t(\d\gamma),
		\end{align*}
		that is
		\[
			\mathcal{Q}^{\lambda,\epsilon}_t (\d\gamma) = \frac
			{\d\lambda}{\d\mu_0^\epsilon}(\pi_0 \gamma) \,  \mathcal{Q}^{ \mu_0^\epsilon,\epsilon}_t(\d\gamma).
		\]
		Now, by the celebrated result of Kolmogorov~\cite{Ko37}, $\mathcal{Q}^{ \mu_0^\epsilon,\epsilon}_t = \mathcal{Q}^{ \mu_0^\epsilon,\epsilon}_t \circ \Theta_t^{-1}$ so that
		\[
			\frac{\d\mathcal{Q}_t^{\lambda,\epsilon}}{\d(\mathcal{Q}_t^{\lambda,\epsilon} \circ \Theta_t)}
			= \frac{\d\mathcal{Q}_t^{\lambda,\epsilon}}{\d\mathcal{Q}_t^{\mu_0^\epsilon,\epsilon}} \times
			\Big(\frac{\d\mathcal{Q}_t^{\mu_0^\epsilon,\epsilon}}{\d\mathcal{Q}_t^{\lambda,\epsilon}} \circ \Theta_t \Big)
		\]
		and we conclude the proof using the identity $\pi_0 (\Theta_t \gamma) = \pi_t \gamma$.
	\end{proof}


\begin{thebibliography}{AGG{\etalchar{+}}86}

\bibitem[ABG19]{ABG19}
A.~Arapostathis, A.~Biswas, and D.~Ganguly.
\newblock Certain {L}iouville properties of eigenfunctions of elliptic
  operators.
\newblock {\em Trans. Amer. Math. Soc.}, 371(6):4377--4409, 2019.

\bibitem[AGG{\etalchar{+}}]{AGG+}
W.~Arendt, A.~Grabosch, G.~Greiner, U.~Moustakas, R.~Nagel, U.~Schlotterbeck,
  U.~Groh, H.~P. Lotz, and F.~Neubrander.
\newblock {\em One-parameter semigroups of positive operators}, volume 1184 of
  {\em Lecture Notes in Mathematics}.
\newblock Springer, 1986.

\bibitem[A82]{An82}
B.~D.~O. Anderson.
\newblock Reverse-time diffusion equation models.
\newblock {\em Stoch. Process. Appl.}, 12(3):313--326, 1982.

\bibitem[BCFG18]{BCFG18}
L.~Bertini, R.~Chetrite, A.~Faggionato, and D.~Gabrielli.
\newblock Level 2.5 large deviations for continuous-time {M}arkov chains with
  time periodic rates.
\newblock {\em Ann. Henri Poincar{\'e}}, 19(10):3197--3238, 2018.

\bibitem[BCX21]{BCX20}
A.~Budhiraja, Y.~Chen, and L.~Xu.
\newblock Large deviations of the entropy production rate for a class of
  gaussian processes.
\newblock {\em J. Math. Phys.}, 62(5):052702, 2021.

\bibitem[BDG15]{BD15}
L.~Bertini and G.~Di~Ges{\`u}.
\newblock Small noise asymptotic of the {G}allavotti--{C}ohen functional for
  diffusion processes.
\newblock {\em ALEA, Lat. Am. J. Probab. Math. Stat.}, 12:743--763, 2015.

\bibitem[BNV94]{BNS94}
H.~Berestycki, L.~Nirenberg, and S.~R.~S. Varadhan.
\newblock The principal eigenvalue and maximum principle for second-order
  elliptic operators in general domains.
\newblock {\em Commun. Pure Appl. Math.}, 47(1):47--92, 1994.

\bibitem[CJPS19]{CJPS19}
N.~Cuneo, V.~Jak{\v{s}}i{\'c}, C.-A. Pillet, and A.~Shirikyan.
\newblock Large deviations and fluctuation theorem for selectively decoupled
  measures on shift spaces.
\newblock {\em Rev. Math. Phys.}, 31(10):1950036, 2019.

\bibitem[CJPS]{EPbrief}
N.~Cuneo, V.~Jak{\v{s}}i{\'{c}}, C.-A. Pillet, and A.~Shirikyan.
\newblock {\em What is a Fluctuation Theorem?}
\newblock SpringerBriefs in Mathematical Physics. Springer, to be published,
  2023.

\bibitem[DV75]{DV75a}
M.~D. Donsker and S.~R.~S. Varadhan.
\newblock On a variational formula for the principal eigenvalue for operators
  with maximum principle.
\newblock {\em Proc. Nat. Acad. Sci. USA}, 72(3):780--783, 1975.

\bibitem[Fel]{Fel2}
W.~Feller.
\newblock {\em An {I}ntroduction to probability theory and its applications,
  {V}ol.~{II}}.
\newblock Wiley series in probability and statistics. John Wiley~\&~Sons, 1966.

\bibitem[FS97]{FS97}
W.~H. Fleming and S.-J. Sheu.
\newblock Asymptotics for the principal eigenvalue and eigenfunction of a
  nearly first-order operator with large potential.
\newblock {\em Ann. Prob.}, 25(4):1953--1994, 1997.

\bibitem[HP86]{HP86}
U.~G. Haussmann and E.~Pardoux.
\newblock Time reversal of diffusions.
\newblock {\em Ann. Probab.}, 14(4):1188--1205, 1986.

\bibitem[HPS13]{HPS13}
M.~Hitrik and K.~Pravda-Starov.
\newblock Eigenvalues and subelliptic estimates for non-selfadjoint
  semiclassical operators with double characteristics.
\newblock {\em Ann. Inst. Fourier}, 63(3):985--1032, 2013.

\bibitem[HS84]{HS84}
B.~Helffer and J.~Sj\"{o}strand.
\newblock Multiple wells in the semi-classical limit {I}.
\newblock {\em Commun. Partial Differ. Equ.}, 9(4):337--408, 1984.

\bibitem[HSS05]{HSS05}
F.~H{\'e}rau, J.~Sj\"{o}strand, and C.~C. Stolk.
\newblock Semiclassical analysis for the {K}ramers--{F}okker--{P}lanck
  equation.
\newblock {\em Commun. Partial Differ. Equ.}, 30(5-6):689--760, 2005.

\bibitem[JOPP11]{JOPP}
V.~Jak{\v{s}}i{\'c}, Y.~Ogata, Y.~Pautrat, and C.-A. Pillet.
\newblock Entropic fluctuations in quantum statistical mechanics an
  introduction.
\newblock In J.~Frohlich, M.~Salmhofer, V.~Mastropietro, W.~De~Roeck, and L.~F.
  Cugliandolo, editors, {\em Quantum Theory from Small to Large Scales},
  volume~95 of {\em Lecture Notes of the Les Houches Summer School}, pages
  213--410. Oxford University Press, 2011.

\bibitem[JOPS12]{JOPS12}
V.~Jak{\v{s}}i{\'c}, Y.~Ogata, C.-A. Pillet, and R.~Seiringer.
\newblock Quantum hypothesis testing and non-equilibrium statistical mechanics.
\newblock {\em Rev. Math. Phys.}, 24(06):1230002, 2012.

\bibitem[JPS17]{JPS17}
V.~Jak{\v{s}}i{\'{c}}, C.-A. Pillet, and A.~Shirikyan.
\newblock Entropic fluctuations in thermally driven harmonic networks.
\newblock {\em J. Stat. Phys.}, 166(3):926--1015, 2017.

\bibitem[Kat]{Kat}
T.~Kato.
\newblock {\em Perturbation theory for linear operators}, volume 132 of {\em
  Grundlehren der mathematischen Wissenschaften}.
\newblock Springer, second edition, 1995.

\bibitem[Kha]{Has}
R.~Khasminskii.
\newblock {\em Stochastic stability of differential equations}, volume~66 of
  {\em Stochastic modeling and applied probability}.
\newblock Springer, second edition, 2011.

\bibitem[KKT10]{KKT09}
S.~Kusuoka, K.~Kuwada, and Y.~Tamura.
\newblock Large deviation for stochastic line integrals as ${L}^p$-currents.
\newblock {\em Probab. Theory Relat. Fields}, 147(3):649--674, 2010.

\bibitem[K37]{Ko37}
A.~N. Kolmogoroff.
\newblock Zur {U}mkehrbarkeit der statistischen {N}aturgesetze.
\newblock {\em Math. Ann.}, 113:766--772, 1937.
\newblock Engl.: {O}n the {R}eversibility of the statistical laws of nature.

\bibitem[K98]{Ku98}
J.~Kurchan.
\newblock Fluctuation theorem for stochastic dynamics.
\newblock {\em J. Phys. A}, 31(16):3719, 1998.

\bibitem[K07]{Ku07}
J.~Kurchan.
\newblock {G}allavotti--{C}ohen theorem, chaotic hypothesis and the zero-noise
  limit.
\newblock {\em J. Stat. Phys.}, 128(6):1307--1320, 2007.

\bibitem[Lan]{Lan}
E.~M. Landis.
\newblock {\em Second order equations of elliptic and parabolic type}, volume
  171 of {\em Translations of mathematical monographs}.
\newblock Amer. Math. Soc., 1997.

\bibitem[LaRo]{LaRo}
P.~Lancaster and L.~Rodman.
\newblock {\em The algebraic {R}iccati equation}.
\newblock Calderon Press, 1995.

\bibitem[LS99]{LS99}
J.~L. Lebowitz and H.~Spohn.
\newblock A {G}allavotti--{C}ohen-type symmetry in the large deviation
  functional for stochastic dynamics.
\newblock {\em J. Stat. Phys.}, 95(1--2):333--365, 1999.

\bibitem[MNV03]{MNV03}
C.~Maes, K.~Neto{\v{c}}n{\`y}, and M.~Verschuere.
\newblock Heat conduction networks.
\newblock {\em J. Stat. Phys.}, 111(5-6):1219--1244, 2003.

\bibitem[MPSR05]{MPRS05}
G.~Metafune, J.~Pr{\"u}ss, R.~Schnaubelt, and A.~Rhandi.
\newblock {$L^p$}-regularity for elliptic operators with unbounded
  coefficients.
\newblock {\em Adv. Differ. Equ.}, 10(10):1131--1164, 2005.

\bibitem[NP92]{NP92}
R.~D. Nussbaum and Y.~Pinchover.
\newblock On variational principles for the generalized principal eigenvalue of
  second order elliptic operators and some applications.
\newblock {\em J. Anal. Math.}, 59(1):161--177, 1992.

\bibitem[Pro]{Pro}
P.~E. Protter.
\newblock {\em Stochastic integration and differential equations}, volume~21 of
  {\em Stochastic modeling and applied probability}.
\newblock Springer, second edition, 2005.

\bibitem[PW66]{PW66}
M.~H. Protter and H.~F. Weinberger.
\newblock On the spectrum of general second order operators.
\newblock {\em Bull. Amer. Math. Soc.}, 72(2):251--255, 1966.

\bibitem[RBT00]{RT00}
L.~Rey-Bellet and L.~E. Thomas.
\newblock Asymptotic behavior of thermal nonequilibrium steady states for a
  driven chain of anharmonic oscillators.
\newblock {\em Commun. Math. Phys.}, 215(1):1--24, 2000.

\bibitem[Roc]{Roc}
R.~T. Rockafellar.
\newblock {\em Convex analysis}.
\newblock Princeton University Press, 1970.

\bibitem[S83]{Si83}
B.~Simon.
\newblock Semiclassical analysis of low lying eigenvalues {I}.
\newblock {\em Ann. Inst. Henri Poincar{\'e}~A: Phys. th{\'e}or.},
  38(3):295--308, 1983.

\bibitem[Sim]{Sim4}
B.~Simon.
\newblock {\em Operator theory}, volume~4 of {\em A Comprehensive course in
  analysis}.
\newblock Amer. Math. Soc., 2015.

\bibitem[S74]{Sj74}
J.~Sj{\"o}strand.
\newblock Parametrices for pseudodifferential operators with multiple
  characteristics.
\newblock {\em Ark. Mat.}, 12(1):85--130, 1974.

\bibitem[VF70]{FW70}
A.~D. Ventcel' and M.~I. Freidlin.
\newblock On small random perturbations of dynamical systems.
\newblock {\em Russ. Math. Surv.}, 25(1):1--55, 1970.

\bibitem[vZC03]{CVZ03}
R.~van Zon and E.~G.~D. Cohen.
\newblock Extension of the fluctuation theorem.
\newblock {\em Phys. Rev. Lett.}, 91(11):110601, 2003.

\bibitem[WXX16]{WXX16}
F.-Y. Wang, J.~Xiong, and L.~Xu.
\newblock Asymptotics of sample entropy production rate for stochastic
  differential equations.
\newblock {\em J. Stat. Phys.}, 163(5):1211--1234, 2016.

\end{thebibliography}

\newcommand{\etalchar}[1]{$^{#1}$}

\end{document}